\documentclass[12pt]{article}
\usepackage[margin=1in]{geometry}
\usepackage{enumitem,setspace}
\usepackage{amsthm}

\usepackage[pdftitle={The Economics of Equilibrium with Indivisible Goods},pdfauthor={Ravi Jagadeesan and Alexander Teytelboym}]{hyperref}

\usepackage[round,compress]{natbib}
\newcommand\citeposs[1]{\citeauthor{#1}'s \citeyearpar{#1}}
\newcommand\citepossall[1]{\citeauthor*{#1}'s \citeyearpar{#1}}

\usepackage{amsmath,amsthm,amssymb}
\usepackage{mathrsfs}
\usepackage{bbm,comment}
\usepackage{tikz-cd}

\newlength{\tempparskip}
\newlength{\halfparskip}
\usepackage{calc}

\newtheorem{thm}{Theorem}
\newtheorem{prop}{Proposition}
\newtheorem{fact}{Fact}
\newtheorem{lemma}{Lemma}
\newtheorem{claim}{Claim}
\newtheorem{innercustomgeneric}{\customgenericname}
\providecommand{\customgenericname}{}
\newcommand{\newcustomtheorem}[2]{%
  \newenvironment{#1}[1]
  {\renewcommand\customgenericname{#2}%
   \renewcommand\theinnercustomgeneric{\ref*{##1}$'$}%
   \innercustomgeneric
  }
  {\endinnercustomgeneric}
}
\newcustomtheorem{primeprop}{Proposition}
\newcustomtheorem{primethm}{Theorem}

\theoremstyle{definition}
\newtheorem{defn}{Definition}
\newtheorem{innercustomgenericdef}{\customgenericname}
\providecommand{\customgenericname}{}
\newcommand{\newcustomdef}[2]{%
  \newenvironment{#1}[1]
  {\renewcommand\customgenericname{#2}%
   \renewcommand\theinnercustomgenericdef{\ref*{##1}$'$}%
   \innercustomgenericdef
  }
  {\endinnercustomgenericdef}
}
\newcustomdef{primedefn}{Definition}

\usepackage{thmtools}
\declaretheoremstyle[
bodyfont=\normalfont,
]{mystyle}
\declaretheorem[name=Example,style=mystyle]{eg}
\renewcommand{\thmcontinues}[1]{continued}

\newcommand\ssm{\smallsetminus}

\newcommand\valFn{V^j}
\newcommand\val[1]{\valFn(#1)}
\newcommand\dQLFn{D^j}
\newcommand\dQL[1]{\dQLFn(#1)}

\newcommand\utilFn{U^j}
\newcommand\util[1]{\utilFn(#1)}
\newcommand\dHFn{\dQLFn_{\mathrm{H}}}
\newcommand\dH[2]{\dHFn(#1;#2)}

\newcommand\valHFn{\valFn_{\mathrm{H}}}
\newcommand\valH[2]{\valHFn(#1;#2)}

\newcommand\dBFn{\widetilde{D}^j}
\newcommand\dHBFn{\dBFn_{\mathrm{H}}}
\newcommand\dHB[3]{\dHBFn(#1;#2,#3)}
\newcommand\dB[2]{\dBFn(#1;#2)}
\newcommand\e[1]{\mathbf{e}^{#1}}
\newcommand\zero{\mathbf{0}}
\newcommand\p{\mathbf{p}}
\newcommand\tp{\widetilde{\mathbf{p}}}
\newcommand\pComp[1]{p_{#1}}
\newcommand\tpComp[1]{\tilde{p}_{#1}}
\newcommand\pnewi{\pprComp{i},\p_{I \ssm \{i\}}}
\newcommand\tpnewb{\tpprComp{i},\tp_{\B \ssm \{\b\}}}

\newcommand\ppr{\p'}
\newcommand\pprComp[1]{p'_{#1}}

\newcommand\tpprComp[1]{\tilde{p}'_{#1}}

\newcommand\hp{\hat{\p}}
\newcommand\hpComp[1]{\hat{p}_{#1}}

\newcommand\bun{\mathbf{x}}
\newcommand\bunComp[1]{x_{#1}}
\newcommand\bunpr{\bun'}
\newcommand\bunprComp[1]{x'_{#1}}
\newcommand\hbun{\hat{\bun}}
\newcommand\hbunComp[1]{\hat{x}_{#1}}

\newcommand\dvec{\mathbf{d}}
\newcommand\dvecComp[1]{d_{#1}}
\newcommand\dvecpr{\dvec'}
\newcommand\dvecprComp[1]{\dvecComp{#1}'}
\newcommand\bundow{\mathbf{w}}

\newcommand\buny{{\bf y}}

\newcommand\bunj{\bun^j}
\newcommand\bundowj{\bundow^j}
\newcommand\bunndowjnopar{\numerdowj,\bundowj}
\newcommand\bunndowj{(\bunndowjnopar)}

\newcommand\numerdowj{w^j_0}

\DeclareMathOperator*\argmax{arg\,max}
\DeclareMathOperator*\argmin{arg\,min}
\DeclareMathOperator\Conv{Conv}

\newcommand\D{\mathcal{D}}

\newcommand\B{\mathcal{B}}
\newcommand\cvec{\mathbf{c}}
\newcommand\cvecComp[1]{c_{#1}}

\title{The Economics of Equilibrium with Indivisible Goods%
\footnote{%
Part of this work was supported by the National Science Foundation under Grant No.~DMS-1928930 and by the Alfred P. Sloan Foundation under grant G-2021-16778, while the authors were in residence at the Simons Laufer Mathematical Sciences Institute (formerly MSRI) in Berkeley, California, during the Fall 2023 semester.
We thank Fuhito Kojima, Paul Milgrom, and several seminar audiences for helpful comments.
Teytelboym is grateful to Elizabeth Baldwin for discussions about unimodularity.
}
}
\author{Ravi Jagadeesan%
\footnote{Department of Economics, Stanford University.  Email address: \texttt{ravi.jagadeesan@gmail.com}.
}
\and Alexander Teytelboym%
\footnote{Department of Economics, Institute for New Economic Thinking, and St.~Catherine's College, University of Oxford.  Email address: {\tt alexander.teytelboym@economics.ox.ac.uk}.
Teytelboym has received funding from the European Research Council (ERC) under the European Union's Horizon 2020 research and innovation programme (grant agreement No.~949699).
}
}
\date{%
\today}

\begin{document}

\maketitle

\begin{abstract}
This paper develops a theory of competitive equilibrium with indivisible goods
based entirely on economic conditions on %
demand.
The key idea is to analyze complementarity and substitutability between %
bundles of goods, rather than merely between goods themselves.
This approach allows us to formulate sufficient, and essentially necessary, conditions for equilibrium existence,
which unify
settings with complements and settings with substitutes.
Our analysis has implications for auction design.
\end{abstract}

\onehalfspacing

\section{Introduction}

Much of economic analysis relies on goods being perfectly divisible.
For example, standard convexity conditions---which are critical to ensuring the existence of market-clearing prices \citep{ArDe:54,mckenzie1954equilibrium}---are incompatible with indivisibilities.

Allowing for indivisibilities is an issue of economic importance due to their relevance for market design.
For example, an underlying goal behind many auction designs is to guide markets toward competitive equilibrium outcomes.%
\footnote{See, e.g., \cite{Milg:00,milgrom2009assignment}, \cite{AuMi:02}, \cite*{ausubel2006clock}, \cite{klemperer2010product}, and \cite{milgrom2020clock}.}
Moreover, some prominent auction formats directly rely on equilibrium existence as part of their protocols (see, e.g., \cite{milgrom2009assignment} and \cite{klemperer2010product}).
Developing an economic theory of equilibrium with indivisible goods is therefore an issue not only of theoretical, but also of practical, importance.

One celebrated case for which there is a well-behaved theory of auctions and markets with indivisible goods is the case of gross substitutes.
With divisible goods,
gross substitutes provides a condition under which competitive equilibria are stable and t\^atonnement price dynamics converge to equilibrium \citep*{arrow1958stability,arrow1959stability}.
With indivisible goods,
\cite{KeCr:82} showed that
gross substitutes also provides a sufficient condition for the existence of competitive equilibria. %
This theoretical result underpins practical auction designs for substitute indivisible goods \citep{GuSt:00,Milg:00,milgrom2009assignment,AuMi:02,%
klemperer2010product}.
But the gross substitutes condition is also rather restrictive because it entirely rules out complementarities.

If the market contains only two indivisible goods,
the case of complements is isomorphic to the case of substitutes.
However, with more than two indivisible goods,
the case of complements is more delicate.
For example, equilibrium does not generally exist when all goods are complements \citep{BiMa:97},
yet certain patterns of complementarities are compatible with the existence of equilibrium \citep*{greenberg1986strong,danilov2013equilibria}.
As a result,
the equilibrium existence problem has been difficult to penetrate beyond cases with a close economic connection to substitutes.%
\footnote{\label{fn:inconsistentGoods}See, e.g., \cite{greenberg1986strong},
\cite{SuYa:06,SuYa:09}, %
\citet*{danilov2013equilibria}, and \citet{HaKoNiOsWe:11}.
A notable exception is \cite{rostek2018matching}, who showed equilibria exist in matching markets in which all contracts are complementary.
However, analogues of this result for exchange economies require %
including brokers in the economy
in a particular way \citep{rostek2020complementarity}.
}

The goal of the current paper is to develop an economic theory of equilibrium for settings with indivisible goods that allows for complementarities.
Our main result provides
a condition
on the structure of complementarity and substitutability that is sufficient, and essentially necessary,
for the existence of competitive equilibria in markets for indivisible goods.
To formulate this result,
we analyze complementarity and substitutability between bundles of goods, rather than merely between goods themselves.

For ease of exposition,
we first focus on the case in which agents demand at most one unit of each good.
Recall that an essentially necessary condition for guaranteeing equilibrium existence is that no two goods be substitutes to one agent and complements to another---that is, that each pair of goods either be consistently substitutes or consistently complements.%
\footnote{This observation was made in the quasilinear case to provide a counterexample to equilibrium existence by \cite{KeCr:82} and to provide a maximal domain result by \cite{GuSt:99}, and extended to the case with income effects by \cite*{baldwin2021equilibrium};
see also \citet[Section 3.3]{henry1970indivisibilites} and \cite{yang2017maximal}.
As we formalize in Section~\ref{sec:necessityBinary}, %
this condition is only necessary to guarantee equilibrium existence for domains that are \emph{invariant} in the sense that they include all additive valuations and are closed under addition of additive valuations.
}
For example, if Lex regards apples and bananas as perfect complements and values the basket of fruit at $\$1$,
and Vincent would like to buy either piece of fruit for $\$1$,
then it is impossible to clear the market if Claude has one apple and one banana available for sale and is unwilling to cut the fruit \citep{henry1970indivisibilites,KeCr:82}.

This consistency condition is not sufficient to guarantee the existence of equilibrium (beyond the case of substitutes).
For example,
it is always satisfied if all goods are complements.

Our key conceptual insight is to analyze the comparative statics of demand for bundles of goods %
and impose an analogous consistency condition on complementarity and substitutability.
More precisely, we introduce a \emph{bundle consistency} condition
that requires that each pair of appropriate bundles of goods be either consistently complementary or consistently substitutable.
We then show that as equilibrium existence does not depend on how goods are bundled,
bundle consistency is an essentially necessary condition for equilibrium~existence.

Our main technical contribution is to show that,
more surprisingly, bundle consistency is also a sufficient condition for equilibrium existence.
Thus, once bundles are also considered, the \emph{only} essential obstruction to equilibrium existence with indivisible goods is the inconsistency between substitutability and complementarity,
illustrated by the example with Lex, Vincent, and Claude above.
By contrast,
\citeposs{KeCr:82} substitutability property is not the only  fundamental reason for equilibrium existence:
having a family of substitute bundles is sufficient for equilibrium existence, %
but not necessary.%
\footnote{See \citet*{danilov2013equilibria} and \cite{BaKl:19}.
\cite{greenberg1986strong}, \cite{SuYa:06}, and \citet{HaKoNiOsWe:11} developed important domains for equilibrium existence that can be derived by identifying appropriate families of substitute bundles.
}

The observation that underpins our results is that the structure of complementarity and substitutability across bundles of goods can expose obstructions to equilibrium existence that are not apparent from the structure of complementarity and substitutability between goods themselves.
In particular, our bundle consistency condition has bite even if all goods are (consistently) complements.
The following example illustrates these points. %

\begin{eg}[name=Bundle Inconsistency with a Three-Cycle of Complements, label=eg:threecycle]
Suppose that there are three goods---apples, bananas, and coconuts---and that agents see pairs of goods as complements.
It turns out that competitive equilibria do not generally exist in this setting \citep[page 390]{BiMa:97}.
Even though there is no inconsistency between substitutability and complementarity for any pair of goods,
an inconsistency emerges when goods are bundled.
Indeed, suppose that apples and bananas are bundled,
but apples (and coconuts) can also be bought and sold alone.
To buy a banana without an apple, an agent would buy the bundle and sell an apple.
Due to the complementarity between bananas and coconuts, an increase in the price of coconuts can lower demand for bananas.
With bundling, this leads to selling fewer apples, thereby raising excess demand for apples and exposing a substitutability between coconuts and apples.%
\footnote{The general principle at play here is that,
even with divisible goods,
whether a pair of goods (or bundles) is substitutes or complements depends on the choice of definition of other goods or bundles \citep{weinstein2022direct}.
}
To see a bundle inconsistency, note that apples and coconuts also remain complements (for reasons unrelated to bananas).
\end{eg}

Other patterns of complementarities are bundle-consistent;
these patterns turn out to be compatible with equilibrium existence.
The following variation on Example~\ref{eg:threecycle} illustrates.

\begin{eg}[Bundle Consistency in Consecutive Games]
\label{eg:consecutive}
Suppose that apples and bananas are complements,
as are bananas and coconuts,
but that unlike in Example~\ref{eg:consecutive},
coconuts are only complementary to apples in combination with bananas.
This pattern of complementarity is compatible with equilibrium existence \citep{greenberg1986strong}.
And indeed, this pattern of complementarity is bundle-consistent.
As in Example~\ref{eg:threecycle},
if apples and bananas are bundled and apples can be bought and sold alone,
then coconuts and apples still become substitutes.
However, this substitutability does not create an inconsistency of complementarity and substitutability between apples and coconuts:
bundling apples and bananas makes coconuts become complementary only to the apple-banana bundle---rather than to apples.  %
\end{eg}

The difference between the patterns of complementarity between goods across Examples~\ref{eg:threecycle} and~\ref{eg:consecutive} is subtle---%
whether coconuts are complementary directly to apples, or to apples only in conjunction with bananas.
In particular, all pairs of goods are complements in both examples.
Considering the structure of complementarity and substitutability when goods are bundled
makes the economic distinction between the two examples transparent.

A conceptual challenge in formalizing the concept of bundle consistency is to identify \emph{which} bundles that agents must see as consistently substitutable or complementary. %
When goods themselves are complementary,
the relevant bundles are the ones consisting of sets of goods that some agent sees as (strict) complements.
Moreover,
when two goods are substitutes,
there is an underlying complementarity between either good and the opportunity to sell the other good \citep{Ostr:08,HaKoNiOsWe:11,HaKoNiOsWe:13}.
In light of this ``hidden complementarity,''
we consider bundles that include opportunities to sell. %
The relevant bundles then consist of goods and opportunities to sell that some agent sees as (strict) complements.
With indivisible goods,
we show that price effects can capture which bundles consist of strict complements;
this logic applies equally well to bundles that include opportunities to sell.
Our analysis %
therefore %
unifies settings with substitutes with those with complements.

Our results extend to the more general case in which agents can demand multiple units of each good.
In that case,
in addition to inconsistencies between substitutability and complementarity \emph{across goods or bundles},
it is possible that there be inconsistency between substitutability and complementarity \emph{across units of the same good}.
Intuitively,
units of the same good are mechanically substitutes,
so complementarities between units of the same good lead to inconsistencies between substitutability and complementarity at the unit level.
We therefore require that units of the same good be substitutes for each other---%
though they may be complementary to (units of) other goods.
Our \emph{unit consistency} condition %
is closely related to the condition in standard equilibrium theory that there be no increasing returns to scale:
with one indivisible good,
increasing returns arise if and only if units of the good can be complements.
We show that unit and bundle consistency are together sufficient, and essentially necessary, for equilibrium existence with multiunit demand.

\paragraph*{Related literature.}
Our results have a mathematical connection
to geometric approaches to the existence of equilibrium with indivisible goods.
\cite*{DaKoMu:01} developed a theory of convexity for settings with indivisible goods by using methods from combinatorial geometry. %
In their framework,
each set $\mathcal{D}$ of integer vectors satisfying~a~condition called ``total unimodularity''
gives a ``class of discrete convexity,''
defined as the %
preferences for which the edge directions of the convex hulls of (Hicksian) demand sets are elements of $\mathcal{D}$.%
\footnote{\label{fn:introDKMpseudo}These demand sets are also required to include all integer points in their convex hulls.}
They showed each class of discrete convexity gives a domain for equilibrium existence.%

For the quasilinear case,
\cite{BaKl:19} took a dual geometric approach to constructing  \citepossall{DaKoMu:01} domains for equilibrium existence.
\cite{BaKl:19} represented preferences in terms of
the set of price vectors at which demand is nonunique. %
Their key insight was to observe that for quasilinear preferences over indivisible goods and money, %
this set
is a space of the type studied in the mathematical area of tropical geometry.
They then classified valuations into ``demand types'': %
each demand type is defined by a set $\D$
of integer vectors %
that places mathematically natural constraints on the tropical geometric representation.
They showed that when $\D$ is totally unimodular, the class of %
valuations of demand type $\D$ recovers the quasilinear case of the class of discrete convexity corresponding to $\D$;
the class of %
valuations of each (totally) unimodular demand type thus gives a domain for equilibrium existence.%
\footnote{To recover classes of discrete convexity and obtain equilibrium existence, valuations must also be \emph{pseudoconcave} in that demand sets include all integer points in their convex hulls.  See also Footnote~\ref{fn:introDKMpseudo}.}
\citet*{baldwin2020equilibrium} extended part of \citeposs{BaKl:19} work to settings with income effects %
and recovered
classes of discrete convexity beyond~the~quasilinear~case.

\cite*{DaKoLa:2003,danilov2013equilibria} showed that there are classes of discrete convexity corresponding to well-understood domains for equilibrium existence with transferable utility, including the case of substitutes.%
\footnote{%
\cite*{DaKoLa:2003}
showed that there is a unimodular demand type corresponding to the class of (strong) substitutes valuations \citep{KeCr:82,GuSt:99,MiSt:09}; see also \cite{shioura2015gross}.
\cite*{danilov2013equilibria} showed that there are unimodular demand types corresponding to the class of ``(gross) substitutes and complements'' valuations \citep{SuYa:06,SuYa:09},
and to the class of valuations for ``consecutive games'' \citep{greenberg1986strong}; see also \cite{BaKl:19}.
}
Most classes of discrete convexity, and the corresponding totally unimodular demand types, introduced novel domains.
However,
the economic content of these new domains for equilibrium existence has remained elusive.%
\footnote{%
\citet[page 868]{BaKl:19} asserted, based on their Proposition 3.3,
that ``a demand type is defined by a list of vectors that give the possible ways in which [...]~demand can change in response to a small generic price change''
(see also \citet*[page 24]{baldwin2020equilibrium} for a similar assertion for settings with income effects).
However, this assertion is incorrect even in the quasilinear setting, for reasons that Footnote~\ref{fn:falseClaimBK19} in Appendix~\ref{app:bk19} will explain.
}
Our approach defines mathematically equivalent domains to \cite*{DaKoMu:01},
but in terms of economic conditions on the structure of complementarity and substitutability across bundles,
rather than relying on mathematical concepts from combinatorial, convex, or tropical geometry.
We discuss the relationship in more detail in Appendix~\ref{app:connectionsToGeo}. %

Our focus on demand for bundles of goods, rather than goods for themselves,
has a conceptual relationship to three recent analyses of preferences for divisible goods.
\cite{weinstein2022direct} introduced a critique of the standard definitions of substitutes and complements based on their dependence on exactly how goods are bundled.
This issue arises in our setting,
where how all goods are bundled can affect whether two goods are substitutes or complements.
We address this issue by focusing on a natural family of feasible bundles.
\cite{galeotti2022taxes} used bundles corresponding to eigenvectors of the Slutsky matrix to investigate optimal tax and subsidy policy.
In our setting,
there is no analogue of the Slutsky matrix 
due to the presence of indivisibilities.
\cite{rostek2023reallocative} showed that it is possible to design a core-selecting reallocative auction if there is a way to bundle goods to make bundles substitutes for (and efficiently allocated in nonnegative quantities to) all bidders;
they also give a sufficient condition for this property to hold in terms of Slutsky matrices.
In our setting,
the ability to bundle goods in a way that restores substitutability is sufficient for equilibrium existence 
but not necessary %
\citep*{danilov2013equilibria,BaKl:19};
bundle consistency is a strictly weaker condition. %

Subsequent to our original work,
\cite{nguyen2024equilibrium} have applied our approach to developing essentially necessary conditions for equilibrium existence in transferable utility economies to relate the equilibrium existence problem for transferable utility economies to the implementability of random allocations in pseudomarkets.

\section{Setting}
\label{sec:geometry}

The setup follows \cite*{baldwin2021equilibrium}.

There is a set $I$ of indivisible goods
and a set $J$ of agents.
Agents have preferences over consumption vectors $\bun$ of indivisible goods (or bads) and a continuously divisible num\'eraire, which we call money.
Formally, each agent $j$ has a non-empty domain $X^j \subseteq \{0,1,\ldots,M\}^I$ (where $M$ is a positive integer) of feasible consumption vectors of indivisible goods,
and a utility function $\utilFn = \util{x_0,\bun}: \mathbb{R} \times X^j \to \mathbb{R}$,
which we assume is continuous, strictly increasing in the quantity $x_0$ of money, and satisfies
\[\lim_{x_0 \to -\infty} \util{x_0,\bun} = -\infty \quad \text{and} \quad \lim_{x_0 \to \infty} \util{x_0,\bun} = \infty \qquad \text{for all } \bun \in X.%
\footnote{%
Utility can obviously take range in any open interval instead. %
Here, agents are allowed to borrow money, but would never choose to borrow and spend infinitely much money.  All the results continue to hold if agents cannot borrow and find it essential to end up with a positive amount of money (as in \cite{henry1970indivisibilites} and \citet*{baldwin2021equilibrium}).}
\]

Each agent $j$ is endowed with an amount $\numerdowj$ of money and an feasible consumption vector $\bundowj$.
Given an endowment $\bunndowj \in \mathbb{R} \times X^j$ for each agent $j$,
a \emph{competitive equilibrium} consists of a price vector $\p \in \mathbb{R}^I$ and a profile of consumption vectors $(\bunj)_{j \in J}$ such that: (i) 
 $\bunj \in \argmax_{\bun \in X^j} \left\{\util{\numerdowj - \p \cdot (\bun - \bundowj),\bun}\right\}$ for each agent $j$, and (ii)
$\sum_{j \in J} \bunj = \sum_{j \in J} \bundowj.$
Here, $\p \cdot (\bun - \bundowj)$ is the net cost of consuming $\bun$ given an endowment $\bundow$ of goods.

In order to analyze settings both with and without income effects, we focus on (Hicksian) demand
which, following \cite*{baldwin2021equilibrium}, we define by analogy with the standard divisible good setting as the set of feasible bundles that minimize the expenditure of achieving a target utility level at competitive prices.
Formally, define the (Hicksian) \emph{demand} of an agent $j$ at a price vector $\p$ for a utility level $u$
by%
\[\dH{\p}{u} = \argmin_{\bun \in X^j} \left\{\min_{x_0 | \util{x_0,\bun} \ge u} \{x_0 + \p \cdot \bun\}\right\}.
\]

When an agent does not experience income effects, their preferences over goods can be summarized by a \emph{valuation} $\valFn : X^j\to\mathbb{R}$ and their utility function $\util{x_0,\bun}=x_0+ \val{\bun}$ becomes quasilinear.
Without income effects, and agent's demand is simply
\[\dQL{\p} = \argmax_{\bun \in X^j} \{\val{\bun} - \p \cdot \bun\}.\]

Intuitively, a good $i$ is substitutable (resp. complementary) to good $k$ if whenever the price of good $i$ increases, the demand for good $k$ weakly increases (resp. weakly decreases). 
Formally, we say that a
good $i$ is \emph{substitutable} (resp.~\emph{complementary}) to good $k$
if for all agents $j$,
utility levels $u$
and price vectors $\p$ with $\dH{\p}{u} = \{\bun\}$,
and all new prices $\pprComp{i} < \pComp{i}$,
there exists $\bunpr \in \dH{\pnewi}{u}$ such that $\bunprComp{k} \le \bunComp{k}$
(resp.~$\bunprComp{k} \ge \bunComp{k}$).%
\footnote{%
The condition that all goods be substitutes corresponds to \citepossall{baldwin2021equilibrium} ``net substitutes'' condition.
We omit the modifier ``net'' as in standard consumer theory.
}
In Appendix~\ref{app:subsComps}, we show that our definition is equivalent to alternative definitions of substitutability and complementarity that have been proposed in the literature.

\section{Complementarity and substitutability between bundles}\label{sec:bundling}

Rather than only focusing on whether a pair of goods are substitutes or complements, we will need to consider whether bundles of goods are substitutes or complements.

\subsection{Preferences over bundles}

\newcommand\q{\mathbf{q}}
\newcommand\qComp[1]{q_{#1}}
\newcommand\qpr{\mathbf{q}'}
\newcommand\qprComp[1]{q'_{#1}}
\renewcommand\b{\mathbf{b}}
\newcommand\bComp[1]{b_{#1}}
Formally,
a \emph{bundling} is a set $\B \subseteq \{-1,0,1\}^I$ of bundles $\b$ that forms a basis for $\mathbb{R}^I$.

In particular, consumption vectors that include at most one unit of each good are potential bundles.
Moreover, we allow bundles in a bundling to include negative components, which represent sale opportunities included a bundle. 
For example, bundle $\b=(1,-1)$ represents the combination of the first good combined with an opportunity to sell the second good.
Considering bundles that include sale opportunities is important as sale opportunities can be complementary to goods.

Given a (potentially fractional and negative) quantity vector $\q \in \mathbb{R}^\B$ of bundles,
there is a corresponding (potentially fractional and negative)  bundle
$\bun(\q;\B) = \sum_{\b \in \B} \qComp{\b}\b \in \mathbb{R}^I$.
Here, $\qComp{\b}$ represents the quantity of bundle $\b$ included in the quantity vector $\q$. Since bundlings form bases for $\mathbb{R}^I$, any feasible consumption bundle can be uniquely achieved by consuming only bundles in any bundling, though potentially in fractional and/or negative quantities.

We next define demand for bundles when it is bundles, rather than goods, that are priced.
In the case of an agent $j$ with quasilinear preferences,
we can define the \emph{bundled demand}, given a bundling $\B$, at a bundle price vector $\tp \in \mathbb{R}^\B$ by
\[\dB{\tp}{\B} = \argmax_{\q \in \mathbb{R}^\B \mid \bun(\q;\B) \in X^j} \left\{\val{\bun(\q;\B)} - \tp \cdot \q\right\}.\]
Here, bundled demand simply expresses demand in terms of bundles from the given bundling. For example, if an agent's demand at certain prices were $(1,0)$ and the bundling were $\{(1,1),(0,1)\}$, then bundled demand at the corresponding bundle prices would be given by $(q_{(1,1)},q_{(0,1)})=(1,-1)$.
More generally, allowing for income effects, given an agent $j$, a utility level $u$, and a bundling $\B$,
we define \emph{bundled demand}
at bundle price vector $\tp \in \mathbb{R}^\B$~by
\[\dHB{\tp}{u}{\B} = \argmin_{\q \in \mathbb{R}^\B \mid \bun(\q;\B) \in X^j} \left\{\min_{x_0 | \util{x_0,\bun(\q;\B)} \ge u} \{x_0 + \tp \cdot \q\}\right\}.\]

To define substitutability and complementarity between bundles,
analogously to case of goods, we look at the effect of bundled demand by changing the price of one bundle holding the prices of other bundles fixed.
We also consider whether a pair of bundles is consistent in their substitutability or complementarity.

\begin{defn}
Given a bundling $\B$ and bundles $\b,\cvec \in \B$:
\begin{itemize}
\item Bundle $\b$ is \emph{substitutable} (resp.~\emph{complementary}) to bundle $\cvec$
if for all agents $j$,
utility levels $u$,
and price vectors $\tp$ with $\dHB{\tp}{u}{\B} = \{\q\}$,
and all new prices $\tpprComp{\b} < \tpComp{\b}$,
there exists $\qpr \in \dHB{\tpnewb}{u}{\B}$ such that $\qprComp{\b} \le \qComp{\b}$
(resp.~$\qprComp{\b} \ge \qComp{\b}$).%
\footnote{%
\cite{weinstein2022direct} observed that in standard consumer theory with divisible goods, the definitions of substitutes and complements are sensitive to the precise definition (or basis) for goods. %
As our examples below show, whether agents view goods or bundles as substitutes or complements in the indivisible good setting also depends on the bundling.
}
\item Bundle $\b$ is \emph{consistent} with bundle $\cvec$ if $\cvec$ is substitutable to, or complementary, to $\b$.
\end{itemize}
\end{defn}

With two indivisible goods, if each agent can demand at most one unit of each good, consistency is sufficient and essentially necessary for equilibrium existence \citep{henry1970indivisibilites,KeCr:82}.
But with more than two goods,
the situation is more complicated.

\subsection{Leading examples revisited}\label{sec:leading}

We now revisit the example from the introduction that shows that even if all agents view goods as complements, certain bundlings can reveal inconsistencies between substitutability and complementarity
which obstruct equilibrium existence.

\newcommand{\apple}{\mathsf{apple}}
\newcommand{\banana}{\mathsf{banana}}
\newcommand{\coconut}{\mathsf{coconut}}
\newcommand{\date}{\mathsf{date}}
\newcommand{\eld}{\mathsf{elderberry}}
\newcommand{\ga}{\mathsf{a}}

\newcommand{\gb}{\mathsf{b}}

\newcommand{\gc}{\mathsf{c}}

\newcommand{\gd}{\mathsf{d}}

\newcommand{\gel}{\mathsf{e}}

\begin{eg}[continues=eg:threecycle]
Let the set of goods be $I = \{\apple,\banana,\coconut\}$, and let the set of agents be $J=\{1,2,3\}$. Suppose that agents' preferences are quasilinear and are given by the following valuations: 
$$
V^1(\bun)=3\min\{\bunComp{\ga},\bunComp{\gb}\},\quad
V^2(\bun)=3\min\{\bunComp{\gb},\bunComp{\gc}\},\quad
V^3(\bun)=3\min\{\bunComp{\ga},\bunComp{\gc}\}.
$$
Each pair of goods is complementary, so there is no inconsistency between goods.
Yet no competitive equilibrium exists if one unit of each good is available in the economy.%
\footnote{To see this, note that efficiency requires allocating one agent at least two goods that they desire.  Without loss of generality, suppose that agent 1 is allocated $\apple$ and $\banana$ (as well as, possibly, $\coconut$).  Then, we must have $p_{\coconut} = 0$ and $p_{\apple} + p_{\banana} \le 3$. But for agent 2 not to demand $\banana$ and for agent 3 not to demand $\apple$, we must have that $p_{\banana} \geq 3$ and $p_{\apple} \geq 3$, respectively, which is a contradiction.}

Now, let us analyze the comparative statics of bundled demand for all agents. 
Consider the bundle $(1,1,0)$ of $\apple$ and $\banana$ and a corresponding bundling $$\B=\{(1,0,0),(1,1,0),(0,0,1)\}.$$
In bundling $\B$, $\apple$ and $\banana$ are bundled, but $\apple$ and $\coconut$ can be traded separately .

Suppose we start at bundle prices 
$\tp=(1,2,3)$. %
The agents' bundled demands are:
$$\widetilde{D}^1(\tp;\B)=(0,1,0), \quad \widetilde{D}^2(\tp;\B)=(0,0,0), \quad 
\widetilde{D}^3(\tp;\B)=(0,0,0).$$
Indeed, agent 1 wants to consume $\apple$ and $\banana$, so agent 1 simply demands the bundle. Agents 2 and 3 do not wish to buy anything. 

Now suppose that the price of $\coconut$ decreases to 1, so the bundle prices change to $\tp'=(1,2,1)$. 
The agents' bundled demands are:
$$\widetilde{D}^1(\tp';\B)=(0,1,0), \quad \widetilde{D}^2(\tp';\B)=(-1,1,1), \quad 
\widetilde{D}^3(\tp';\B)=(1,0,1).$$
Agent 1's demand (and hence bundled demand) does not change. 
Agent 3 now wishes to buy $\apple$ and  $\coconut$.
But Agent 2 wishes to end up with $\banana$ and $\coconut$, so under the bundling $\B$, agent 2 buys the bundle of $\apple$ and $\banana$ and sells $\apple$ (in addition to buying $\coconut$ separately).
Thus, following the decrease in the price of $\coconut$, agent 2's demand for $\apple$ falls (since he wishes to sell $\apple$) while agent~3's demand for $\apple$ rises.

So, while $\apple$ and  $\coconut$ remain complements for agent 3, these goods have become substitutes for agent 2.
Hence, given the bundling, the bundle that consists of $\apple$ is not consistent with the bundle that consists of $\coconut$.\hfill$\blacksquare$
\end{eg}

In Example~\ref{eg:threecycle}, bundling two complements revealed an inconsistency between two other goods. However, it is not the presence of complements \textit{per se} that creates this bundle inconsistency, but rather the pattern of preferences over complementary goods that is key to a possible inconsistency, which in turn serves as an obstruction to equilibrium existence. The following example makes it clear that bundle consistency (and equilibrium existence) is perfectly compatible with agents' viewing all goods as complements.

\begin{eg}[continues=eg:consecutive]
To formalize this example, let us only replace agent 3's preferences in Example~\ref{eg:threecycle} with the following valuation:
$$
V^3(\bun)=3\min\{\bunComp{\ga},\bunComp{\gb},\bunComp{\gc}\}.
$$

Given the bundling $\B=\{(1,0,0),(1,1,0),(0,0,1)\}$ and bundle prices $\tp=(1,2,3)$ agents bundled demands are as before:
$$\widetilde{D}^1(\tp;\B)=(0,1,0), \quad \widetilde{D}^2(\tp;\B)=(0,0,0), \quad 
\widetilde{D}^3(\tp;\B)=(0,0,0).$$
But following the decrease in the price of $\coconut$ to 1 (making bundle prices $\tp'=(1,2,1)$), agents' bundled demands are:
$$\widetilde{D}^1(\tp';\B)=(0,1,0), \quad \widetilde{D}^2(\tp';\B)=(-1,1,1), \quad 
\widetilde{D}^3(\tp';\B)=(0,1,1).$$
Thus, agent 3 now demands the bundle of $\apple$ and $\banana$ as well as the $\coconut$, but not the bundle consisting of $\apple$ alone.

So, $\apple$ and $\coconut$ become substitutes for all agents.
In particular, the bundle that consists of $\apple$ is consistent with the bundle that consists of $\coconut$ under bundling $\B$.
Moreover, competitive equilibria exist for all endowments \citep{greenberg1986strong}.
\hfill$\blacksquare$
\end{eg}

While the patterns of complementarities between the two preceding examples are only subtly different, there are dramatic differences in the implications for equilibrium existence.
One might hope that in order to reveal all obstructions to equilibrium existence, it is sufficient to consider bundlings in which only complementary goods are bundled.
However, in Online Appendix~\ref{oapp:egs}, we provide an example showing that in order to reveal all obstructions to equilibrium existence, we must also include ``hidden'' complementarities between goods and opportunities to sell other goods in the complements that we bundle. %

The challenge of identifying bundlings relevant for obstructions to equilibrium existence can be entirely addressed by considering bundles that correspond to price effects.

\vspace{-6pt}

\section{Price effects and bundle consistency}
\label{sec:pricebundle}

The goal of this section is to use price effects to identify which bundlings are critical for equilibrium existence: they either contain bundles of (strictly) complementary goods or bundles of  (strict) ``hidden complements'', i.e., goods and opportunities to sell other goods.

Price effects will allow us to measure both types of complementarities: complementarities between goods and ``hidden'' complementarities between goods and sale opportunities.
The definition of (compensated) price effects is analogous to the setting with divisible goods.\footnote{When an agent's utility function is quasilinear, the agent does not experience income effects so the compensated price effect is just the overall price effect.}

\begin{defn}
A \emph{(compensated) price effect for agent $j$ and good $i$} is a nonzero vector $\Delta \bun$
for which there exist
a utility level $u$,
a price vector $\p$, and a new price $\pprComp{i} < \pComp{i}$ with $\dH{\p}{u} = \{\bun\}$ and %
$\bun + \Delta \bun \in \dH{\pnewi}{u}$.
\end{defn}

Consider the change in (compensated) demand $\bunpr - \bun = \Delta \bun$ following a fall in the price of good $i$. By the (compensated) law of demand, the $i\text{th}$ entry of the $\Delta \bun$ is non-negative. If then $k\text{th}$ entry of the $\Delta \bun$ is positive, the price effect reveals that goods $i$ and $k$ are complements. On the other hand, $k\text{th}$ entry of $\Delta \bun$ is negative, then the demand for good $k$ falls, in other words the agent would like to sell the good. In this case, price effect reveals a ``hidden complementarity'' between the good and the opportunity to sell the good. 

As we will show in Section~\ref{sec:priceEffs},
in the indivisible good case,
price effects are a more powerful reflection of the structure of complementarity and substitutability than in the divisible good case.

The following definition summarizes which bundles %
are critical to identifying obstructions to equilibrium existence. 

\begin{defn}%
A bundle $\b \in \{-1,0,1\}^I$ is \emph{relevant} if it is either a price effect $\Delta\bun$ for some agent $j$ or a bundle $\e{i}$ consisting only of good $i$. 
\end{defn}

For example, if all goods are complements (as in Examples~\ref{eg:threecycle} and~\ref{eg:consecutive}) that the relevant bundles only contain bundles of goods. 
If some goods are complements and others are~substitutes, then the relevant bundles contain bundles of goods and opportunities to sell other~goods.
    
We can now introduce the key condition for equilibrium existence with indivisible goods.

\begin{defn}
\label{def:bundleConsistent}
Preferences are \emph{bundle-consistent} if for each bundling $\B$ that consists only of relevant bundles,
each pair of bundles $\b,\cvec \in \B$ are consistent.
\end{defn}

Definition~\ref{def:bundleConsistent} says that any two bundles must be either substitutable or complementary to each other whenever we look at any bundling that consists only of goods and bundles that correspond to price effects (and contain at most one unit of each good or sale opportunity).

\section{When agents demand at most one unit of each good}
\label{sec:binary}

We start by considering the case in which each agent demands at most one unit of each~good.\footnote{This is what \cite{MiSt:09} call the ``binary valuations'' case. It is equivalent to an economy in which every unit of every good is treated as a distinct good.}
Our first result shows that bundle consistency is in fact sufficient for equilibrium existence. %

\begin{thm}
\label{thm:existBinary}
Suppose that each agent demands at most one unit of each good.
If preferences are bundle-consistent,
then competitive equilibria exist. %
\end{thm}

Theorem~\ref{thm:existBinary} gives a simple, economically interpretable condition for equilibrium existence that is expressed in terms of comparative statics of demand.
As we will show, Theorem~\ref{thm:existBinary} encompasses many known existence results, such as ones for substitutes \citep{KeCr:82}, classes of complements \citep*{greenberg1986strong}, and substitutes and complements \citep{SuYa:06,SuYa:09}. %
It (together with its extension to settings which agents can demand multiple units of goods in Section~\ref{sec:mult}) also encompasses equilibrium existence results based on combinatorial and tropical geometric analyses \citep*{DaKoMu:01,danilov2013equilibria,BaKl:19}.

\subsection{Necessity of bundle consistency}
\label{sec:necessityBinary}
We now show that bundle consistency is necessary for equilibrium existence,
in a much stronger sense than previous necessity results in the literature.

In specific economies,
it is possible that some agents view a pair of good as complements, other agents view them as substitutes, and yet equilibrium exists \citep{BiMa:97,Ma:98}.
Thus, to formulate a necessity result,
we require that the preference domain for existence to be sufficiently rich in that it is ``invariant'' under a simple family of transformations.
We focus on the case of transferable utility economies; an analogous result for economies with income effects is in Online Appendix~\ref{oapp:neccIncomeEffs}.
\begin{defn}\label{def:invariance}
A domain $\mathcal{V}$ of valuations is \emph{invariant} if:
\begin{enumerate}[label=(\roman*)]
\item\label{def:invariancei} for all $V \in \mathcal{V}$ and all price vectors $\p \in \mathbb{R}_{\ge 0}^I$, writing $V'(\bun) = V(\bun) + \p \cdot \bun$, we have that $V' \in \mathcal{V}$; and
\item\label{def:invarianceii} $\mathcal{V}$ contains the valuation $V_0: \{0,1\}^I \to \mathbb{R}$ defined by $V_0(\bun) = 0$.
\end{enumerate}
\end{defn}

Part~\ref{def:invariancei} of the definition of invariance requires closure under the addition of positive linear functions. Subsequent work by \citet{nguyen2024equilibrium} offers two sufficient conditions for Part~\ref{def:invariancei}. First, any domain defined by conditions on the demand correspondence at all prices satisfies Part~\ref{def:invariancei}. Second, any domain defined by conditions on how the addition of a good can affect the marginal value of another good satisfies Part~\ref{def:invariancei}.
Examples of domains that satisfy Part~\ref{def:invariancei} thus include many conditions for existence in the literature: substitutes \citep{KeCr:82}, consecutive games \citep{greenberg1986strong}, substitutes and complements \citep{SuYa:06}, and sign-consistent tree valuations \citep{CaOzPa:15}, as well as classes of discrete convexity \citep*{DaKoMu:01} and (totally) unimodular demand types \citep{BaKl:19}. %
On the other hand, domains based on conditions on specific valuation profiles, as proposed by \cite{BiMa:97} and \cite{Ma:98}, do not satisfy Part~\ref{def:invariancei} of invariance.

Furthermore, given any bundle-consistent domain of valuations,
if the domain does not satisfy Part~\ref{def:invariancei} of invariance,
then there is a larger, bundle-consistent domain that satisfies Part~\ref{def:invariancei} of invariance.
The larger domain can be constructed by simply adding nonnegative linear functions to valuations in the original domain.

Part~\ref{def:invarianceii} of the definition of invariance requires that the domain include a trivial valuation at which the agent values all bundles at zero.
Part~\ref{def:invarianceii} gives meaning to ``goods'' in our model:
if it were not satisfied, then agents could never view individual goods independently,
and it would not be clear why the primitive goods were defined correctly (or why each ``good'' should be priced separately).
We therefore view Part~\ref{def:invarianceii} of invariance as economically innocuous.%
\footnote{Note also that Part~\ref{def:invarianceii} of invariance is automatically satisfied if the economy can contain a seller who seeks to efficiently assign goods to agents.}

The following result states that for invariant domains of valuations,
bundle consistency is necessary for equilibrium existence.
This provides a partial converse to Theorem~\ref{thm:existBinary}.

\begin{thm}
\label{thm:neccBinary}
If competitive equilibria exist in all economies in which agents have valuations in an invariant domain $\mathcal{V}$, %
then the valuations in $\mathcal{V}$ are bundle consistent.
\end{thm}

Theorem~\ref{thm:neccBinary} implies that bundle consistency encompasses most previous domains for equilibrium existence.%
\footnote{Theorem~\ref{thm:neccBinary} also provides a formal sense in which inconsistencies between \emph{goods} obstructs equilibrium existence beyond the case of substitutes (see Footnote~\ref{fn:inconsistentGoods}).}
Indeed, bundle consistency subsumes substitutes \citep{KeCr:82}, consecutive games \citep{greenberg1986strong}, and substitutes and complements \citep{SuYa:06}, as well as classes of discrete convexity \citep*{DaKoMu:01} and totally unimodular demand types \citep{BaKl:19}.
Indeed,
these domains all place price-independent restrictions on demand,
and guarantee existence in all exchange economies with preferences in the domain.%
\footnote{%
Note, however, bundle consistency does not subsume equilibrium existence results that rely on special conditions on the joint structure of preferences or endowments in the economy.
For example,
it does not subsume conditions for equilibrium existence developed by \cite{BiMa:97} and \cite{Ma:98} that are based on the entire profile of valuations
(Part~\ref{def:invariancei} of invariance fails for these domains).
Moreover,
it does not subsume \citeposs{CaOzPa:15} existence result for the domain of ``sign-consistent tree valuations'' because that result requires conditions on the total endowment of indivisible goods (see Section 3.3 of their Supplemental Material), and hence does not apply in all exchange economies.
}

\subsection{Interpretation of price effects}
\label{sec:priceEffs}

Why are price effects so powerful in identifying all the complementarities and ``hidden'' complementarities?
While the price effect clearly reveals \emph{pairwise} complementarities between good $i$ and other goods, and ``hidden'' complementarities good $i$ and opportunities to sell other good, it is not obvious whether or not it captures complementarities between \emph{sets} of goods (or both forms of complementarities between goods other than $i$). 
However, as the following proposition shows, in the setting with indivisible goods, price effects contain far more information about the relationships between goods when we observe a simultaneous change in demand of multiple goods.

Formally, we say that
good $i$ is a \emph{strict complement to good $k$ at some prices},
if there exist an agent $j$,
a utility level $u$,
a price vector $\p$
with $\dH{\p}{u} = \{\bun\}$,
a new price $\pprComp{i} < \pComp{i}$,
and a new demand $\bunpr \in \dH{\pnewi}{u}$
such that that $\bunprComp{k} > \bunComp{k}$.\footnote{At the beginning of Section~3, we gave a definitions of complements (and substitutes) \emph{for all price vectors}.}
To define strict ``hidden'' complementarities, note that implicit demand for an opportunity to sell good $k$ is $M-\bunComp{k}$.
We say that there is a \emph{strict hidden complementarity involving good $i$ and good $k$ at some prices} %
if there exist an agent $j$,
a utility level $u$,
a price vector $\p$ with $\dH{\p}{u} = \{\bun\}$,
a new price $\pprComp{i} < \pComp{i}$, and a new demand $\bunpr \in \dH{\pnewi}{u}$
such that $M - \bunprComp{k}> M- \bunComp{k}$.\footnote{Of course, this is just a definition of strict substitutes at some prices, but conceptually our analysis relies on bundles involves opportunities to sell as Example~\ref{eg:hiddencomp}.}

\begin{prop}
\label{prop:transitivityCompBinary}
Suppose that agent $j$ demands at most one unit of each good.  Let
$\Delta \bun$ is a (compensated) price effect for agent $j$ and good $i$,
and let $k,\ell$ be distinct goods.
\begin{enumerate}[label=(\alph*)]
\item If $(\Delta \bun)_k,(\Delta \bun)_\ell > 0$ or $(\Delta \bun)_k,(\Delta \bun)_\ell < 0$,
then $k$ and $\ell$ are strict complements at some prices.
\item If $(\Delta \bun)_k > 0 > (\Delta \bun)_\ell$,
then there is a strict hidden complementarity involving $k$ and $\ell$ at some prices.
\end{enumerate}
\end{prop}

Proposition~\ref{prop:transitivityCompBinary} says that following a price decrease there is a change in demand of two goods in the same (resp. opposite) direction then must be strict (resp.~strict hidden) complementarity between the goods at some prices.
For example, if the price of $\apple$ falls and the demand for $\banana$ and $\coconut$ both rise, then $\banana$ and $\coconut$ must be strict complements at these prices. In classic consumer theory, one would not be able to make such an inference without knowing the effect of changing the price of $\banana$ on demand for $\coconut$.
As a result, Proposition~\ref{prop:transitivityCompBinary} provides a simple interpretation of which bundles are relevant: we bundle goods whose price effects have the same sign (because they must be strict complements) and we bundle goods with opportunities to sell other goods whose prices effects have the opposite sign (because of the hidden complementarity involving the goods).

\section{When agents demand multiple units of any good}
\label{sec:mult}

We now turn to the case in which agents might demand multiple units of any good.\footnote{This is the what \citet{MiSt:09} call the ``multiunit valuations'' case.}
The definitions of price effects, relevant bundles, and bundle consistency carry over verbatim from the Section~\ref{sec:pricebundle}. %
By definition, relevant bundles include at most one unit of each good.

However, stronger conditions are needed for equilibrium existence in the multiunit setting than in the case in which agents demand at most one unit of each good.
The issue is that we need to rule out the possibility of increasing returns scale in units (e.g., minimum viable quantity of units),
as increasing returns can obstruct equilibrium existence for reasons familiar from classic general equilibrium theory.
We will therefore impose a condition called ``unit consistency’’ which will ensure that units of the same good are substitutes for~each~other.

\newcommand\outilFn{\overline{U}^j}
\newcommand\outil[1]{\outilFn(#1)}
\newcommand\odHFn{\odQLFn_{\mathrm{H}}}
\newcommand\odH[2]{\odHFn(#1;#2)}
\newcommand\odQLFn{\overline{D}^j}
\newcommand\odQL[1]{\odQLFn(#1)}

To define unit consistency formally,
we regard each unit of each good as a separate item.
Formally, an \emph{item} consists of a good $i$ and a serial number $1 \le m \le M$.
Hence,
the set of items is $\overline{I} = I \times \{1,2,\ldots,M\}$.
Given a bundle $\overline{\bun} \in \{0,1\}^{\overline{I}}$ of items,
there is a corresponding bundle $\bun = \pi(\overline{\bun})$ of goods defined by
\[\pi(\overline{\bun})_i = \sum_{m = 1}^M \bun_{(i,m)}.\]
Given any utility function $\utilFn$ for goods and money, there is a corresponding utility function $\outilFn: \mathbb{R} \times \pi^{-1}(X^j) \to \mathbb{R}$
for items and money
defined by
\[\outil{x_0,\overline{\bun}} = \util{x_0,\pi(\overline{\bun})}.\]
We can now formally state the definition of unit consistency.
\begin{defn}
Utility function $\utilFn$ is \emph{unit-consistent} if for the utility function $\outilFn$,
for all goods $i$ and serial numbers $1 \le m < m' \le M,$
the items $(i,m)$ and $(i,m')$ are substitutes.%
\footnote{As Lemma~\ref{lem:unitConsistentConsecutiveInteger} in Appendix~\ref{app:proofs} shows,
unit consistency implies a property introduced by \cite{MiSt:09} that at each price vector,
the set of demanded quantities of each good consists of consecutive integers.
In general, unit consistency is a stronger condition than this ``consecutive integer property.''}
\end{defn}

Imposing unit consistency turns out to be sufficient to deal with the additional difficulties of the multiunit demand setting.

\begin{thm}
\label{thm:existGen}
If preferences are unit- and bundle-consistent,
then competitive equilibria exist. %
\end{thm}

Theorem~\ref{thm:existGen} generalizes Theorem~\ref{thm:existBinary} for settings in which agents can demand multiple units of any good.
This extension encompasses an existence result for substitutes due to \citet{MiSt:09}, as well as
existence results of \citet*{DaKoMu:01} and \citet{BaKl:19} (see Appendix~\ref{app:connectionsToGeo}). 

To prove Theorem~\ref{thm:existGen} (of which Theorem~\ref{thm:existBinary} is a special case),
we prove that each domain of preferences that is unit- and bundle-consistent
lies within a ``class of discrete convexity'' in the sense of \citet*{DaKoMu:01} (see Proposition~\ref{prop:dkmRelationshipConsistentToDc} in Appendix~\ref{app:connectionsToGeo} for a precise statement).
The proof of this mathematical connection is fairly involved and proceeds in three steps.
A first step is to show that under unit consistency,
price effects carry mathematical information about the tropical geometric representation of preferences studied by \cite{BaKl:19}.%
\footnote{In the case in which agents demand at most one unit of each good,
this connection was developed by \citet*[Theorem 1]{baldwin2021consumer}.
While \citet*[Theorems 2 and 3]{baldwin2021consumer} offered versions of their Theorem 1 for the multiunit demand case,
these versions involve combinations of price effects rather than price effects themselves,
and are therefore not useful for the current paper's economic approach to equilibrium existence.
}
A second step is a technical result inspired by \cite{schrijver1998theory},
which relates a form of bundle consistency to the total unimodularity property that lies at the heart of discrete convexity.
This second step also turns out to provide a computationally simple way to test for bundle consistency,
which we develop in Section~\ref{sec:testTotalUnimod}.
A third step is to show that unit and bundle consistency imply that demand sets cannot miss integer points in their convex hulls;
here, we rely on a combinatorial geometric result underlying \citeposs{danilov2004discrete} theory of discrete convexity.%
\footnote{In the case in which agents demand at most one unit of each good,
this third step is unneeded.}
We can then conclude equilibrium existence using the equilibrium existence results of \citet*{DaKoMu:01}.
The detailed arguments are in Appendix~\ref{app:proofs}.
Note that while the proof uses combinatorial geometric methods,
these methods are not needed to understand the statement of, or the economic content of, our results.

Our final main result shows that unit and bundle consistency are essentially necessary to guarantee equilibrium existence under multiunit demand, thereby generalizing Theorem~\ref{thm:neccBinary} to the multiunit setting and showing that unit consistency is also essentially necessary for equilibrium existence.

\begin{thm}
\label{thm:neccGen}
If competitive equilibria exist in all economies in which agents have valuations in an invariant domain $\mathcal{V}$,
then the valuations in $\mathcal{V}$ are unit- and bundle-consistent.
\end{thm}

The proof of Theorem~\ref{thm:neccGen} (of which Theorem~\ref{thm:neccBinary} is a special case) shows first the necessity of unit consistency and then the necessity of bundle consistency.
The proof of the necessity of unit consistency follows similar logic to \cite{MiSt:09},
but is more involved due to the possibility of complementarities;
this step is not needed to prove Theorem~\ref{thm:neccBinary}.
The proof of the necessity of bundle consistency is obtained by applying logic similar to \cite{KeCr:82} and \cite{GuSt:99} to bundles, rather than goods.

\subsection{Interpretation of price effects}
We already know that when agents demand at most one unit of each good, price effects pinpoint the structure of complementarities between sets of goods (as well as hidden complementarities).
The following proposition is an analogue of Proposition~\ref{prop:transitivityCompBinary} and shows
that price effects allow us to identify sets of strictly complementary goods and/or sale opportunities in the case when agents demand more than one unit of some goods.

\begin{prop}
\label{prop:transitivityCompGen}
Suppose that agent $j$'s preferences are unit-consistent.
Let $\Delta \bun \in \{-1,0,1\}^I$ be a (compensated) price effect for agent $j$ and good $i$,
and let $k,\ell$ be distinct goods.
\begin{enumerate}[label=(\alph*)]
\item \label{part:transitivityCompGenComp} If $(\Delta \bun)_k,(\Delta \bun)_\ell > 0$ or $(\Delta \bun)_k,(\Delta \bun)_\ell < 0$,
then $k$ and $\ell$ are strict complements at some prices.
\item \label{part:transitivityCompGenSubs} If $(\Delta \bun)_k > 0 > (\Delta \bun)_\ell$,
then there is strict hidden complementarity involving $k$ and $\ell$ at some prices.
\end{enumerate}
\end{prop}

While the intuition for 
Proposition~\ref{prop:transitivityCompGen} is the same as for Proposition~\ref{prop:transitivityCompBinary},
note that
Proposition~\ref{prop:transitivityCompGen} restricts to price effects in which demand changes by at most one unit (as in the definition of relevance).
This additional restriction is crucial.
Indeed, consider a variation on Example~\ref{eg:consecutive}: there is one agent who views an $\apple$ and a $\banana$ as strict complements and $\banana$ and a $\coconut$ as strict complements, but views $\apple$ and $\coconut$ as independent: \[V(\bun)=
\max_{0 \le y \le \bunComp{\gb}} \big\{3 \min\{\bunComp{\ga},y,1\} + 3 \min\{\bunComp{\gb}-y,\bunComp{\gc},1\}\big\}.\]
At prices $\p=(1,3,1)$, the agent demands nothing, i.e., $D^j(\p)=(0,0,0)$.
Now suppose that the price of $\banana$ falls, so $\p'=(1,1,1)$. The agent then demands $\apple$, $\coconut$ and two units of $\banana$, so $D^j(\p')=(1,2,1)$. 
But note that following the fall in the price of $\banana$, demand for $\apple$ and $\coconut$ increased simultaneously,
which would suggest that $\apple$ and $\coconut$ might be strict complements, even though the agent views them as independent goods. 
The reason for the discrepancy is that the hypothesis of Proposition~\ref{prop:transitivityCompGen} is not satisfied in this example, as we considered a price effect in which the demand for $\banana$ changed by two units.
Proposition~\ref{prop:transitivityCompGen} shows that by restricting to price effects in which demand changes by at most one unit, we can still use price effects identify both strict complementarities and strict hidden complementarities, and therefore use price effects to sensibly define relevant bundles.

\section{Testing for bundle consistency}
\label{sec:testTotalUnimod}

While bundle consistency is an economically natural condition,
one may ask whether there is a computationally straightforward way to test whether a given domain of preferences is bundle-consistent.
It turns out that one can test for bundle consistency by checking whether the set of all price effects that lie in $\{-1,0,1\}^I$ is totally unimodular.
Recall that a set $S$ of integer vectors is \emph{totally unimodular} if every square submatrix of the matrix whose columns are the elements of $S$ has determinant 0 or $\pm 1$.%
\footnote{See, e.g., \cite{schrijver1998theory} for efficient tests for, and equivalent characterizations of, total unimodularity.}

\begin{prop}\label{prop:totallyunimodgen}
Suppose preferences are unit-consistent.
Preferences are bundle-consistent if and only if the set of all agents' price effects that lie in $\{-1,0,1\}^I$ is totally unimodular.
\end{prop}

The proposition applies to both the setting of Section~\ref{sec:binary} and the setting of Section~\ref{sec:mult}
as unit consistency is automatic if each agent demands at most one unit of each good.%
\footnote{In that case, the restriction to price effects that lie in $\{-1,0,1\}^I$ is also vacuous.}
In addition to its conceptual value,
the ``only if'' direction comprises one of the steps in the proof of Theorems~\ref{thm:existBinary} and~\ref{thm:existGen} by providing a useful, mathematical consequence of bundle consistency.

The following example illustrates how to apply Proposition~\ref{prop:totallyunimodgen} to verify bundle consistency.

\begin{eg}[\protect{\citealp*[Example 4]{danilov2013equilibria}}]
\label{eg:dkl}
Let the set of goods be $I = \{\apple,\banana,\coconut,\date\}$ and the set of agents be $J=\{1,2,3,4\}$.
Each agent demands at most one unit of each good.

Suppose that agents' preferences are quasilinear and are given by the following valuations: 
$$
V^1(\bun)=\bunComp{\ga}+\bunComp{\gb}+\bunComp{\gc}+\bunComp{\gd}, \quad
V^2(\bun)=3\min\{\bunComp{\ga},\bunComp{\gb}\},\quad
V^3(\bun)=3\min\{\bunComp{\gb},\bunComp{\gc}\},$$
\vspace{-1cm}
$$V^4(\bun)=3\min\{\bunComp{\gc},\bunComp{\gd}\},\quad V^5(\bun)=3\min\{\bunComp{\gd},\bunComp{\ga}\},\quad
V^6(\bun)=3\min\{\bunComp{\ga},\bunComp{\gb},\bunComp{\gc},\bunComp{\gd}\}.
$$
Thus, agent~1 has additive preferences, the next four agents regard pairs of goods as complementary, while agent~5 regards all goods as complementary.

The set of all price effects consists of the columns of the following matrix: $$\begin{pmatrix}
1 & 0 & 0 & 0 & 1 & 0 & 0 & 1 & 1 \\  
0 & 1 & 0 & 0 & 1 & 1 & 0 & 0 & 1 \\  
0 & 0 & 1 & 0 & 0 & 1 & 1 & 0 & 1 \\ 
0 & 0 & 0 & 1 & 0 & 0 & 1 & 1 & 1 
\end{pmatrix}$$
The first four columns describe the price effects for agent~1, the next four columns describe the price effects for agents~2, 3, 4 and 5 and the final column describes the price effects for agent~6.
It is easy to verify that this matrix is totally unimodular.
Therefore, in this example, agents' preferences are bundle-consistent (Proposition~\ref{prop:totallyunimodgen}), so equilibria exist for all endowments under these preferences (Theorem~\ref{thm:existBinary}).
\hfill$\blacksquare$
\end{eg}

In Example~\ref{eg:dkl}, there is \emph{no} bundling for which all agents view the bundles as (strong) substitutes \citep*{danilov2013equilibria}. 
Example~\ref{eg:dkl} therefore illustrates that our equilibrium existence result cannot be deduced from \citeposs{KeCr:82} sufficient conditions by identifying a family of substitute bundles, unlike for the equilibrium existence results of \cite{greenberg1986strong} and \cite{SuYa:06}.

\section{Conclusion}

We showed that with indivisible goods,
the existence of equilibrium is fundamentally determined by whether there is a conflict between complementarity and substitutability of a pair of bundles of goods.
When agents can demand more than one unit of any good,
units of the same good must also not be complementary.
Our analysis provides a unified, economic framework that allows for complementarity, substitutability, and income effects.

Our results have implications for the design of sealed-bid multi-item auctions that guarantee market-clearing.
Currently, parsimonious bidding languages based on ideas from producer theory have been developed only for the case of quasilinear substitutes preferences \citep{milgrom2009assignment}.
Our results suggest that by using consumer theory and the economic concept of bundling,
natural bidding languages can be developed
that would allow bidders to express richer preferences that exhibit complementarities and income effects.%
\footnote{%
\cite{klemperer2010product} took a geometric approach to obtain the bidding languages proposed by \cite{milgrom2009assignment}, focusing on the two-good case for which LIPs can in fact be visualized.
As discussed by \citet[Section 6.7]{BaKl:19},
tropical geometry can be used to design and analyze bidding languages beyond the two-good case \citep*{baldwin2024implementing}, and for multi-item auctions more generally.
The results of this paper imply generally that any outcome that can be implemented by a tropical geometric auction format can instead be achieved using
bidding languages based on producer theory.
Unlike the tropical geometric approach, the consumer theory approach also applies in settings with income effects.
}

\singlespacing
\bibliographystyle{chicago}
\bibliography{bib}

\begin{thebibliography}{}

\bibitem[\protect\citeauthoryear{Arrow, Block, and Hurwicz}{Arrow
  et~al.}{1959}]{arrow1959stability}
Arrow, K.~J., H.~D. Block, and L.~Hurwicz (1959).
\newblock On the stability of the competitive equilibrium, {II}.
\newblock {\em Econometrica\/}~{\em 27\/}(1), 82--109.

\bibitem[\protect\citeauthoryear{Arrow and Debreu}{Arrow and
  Debreu}{1954}]{ArDe:54}
Arrow, K.~J. and G.~Debreu (1954).
\newblock Existence of an equilibrium for a competitive economy.
\newblock {\em Econometrica\/}~{\em 22\/}(3), 265--90.

\bibitem[\protect\citeauthoryear{Arrow and Hurwicz}{Arrow and
  Hurwicz}{1958}]{arrow1958stability}
Arrow, K.~J. and L.~Hurwicz (1958).
\newblock On the stability of the competitive equilibrium, {I}.
\newblock {\em Econometrica\/}~{\em 26\/}(4), 522--52.

\bibitem[\protect\citeauthoryear{Ausubel, Cramton, and Milgrom}{Ausubel
  et~al.}{2006}]{ausubel2006clock}
Ausubel, L.~M., P.~Cramton, and P.~Milgrom (2006).
\newblock The clock-proxy auction: A practical combinatorial auction design.
\newblock In M.~Bichler and J.~K. Goeree (Eds.), {\em Handbook of Spectrum
  Auction Design}, pp.\  120--40. Cambridge University Press.

\bibitem[\protect\citeauthoryear{Ausubel and Milgrom}{Ausubel and
  Milgrom}{2002}]{AuMi:02}
Ausubel, L.~M. and P.~R. Milgrom (2002).
\newblock Ascending auctions with package bidding.
\newblock {\em Frontiers of Theoretical Economics\/}~{\em 1\/}(1), 1--42.

\bibitem[\protect\citeauthoryear{Baldwin, Edhan, Jagadeesan, Klemperer, and
  Teytelboym}{Baldwin et~al.}{2020}]{baldwin2020equilibrium}
Baldwin, E., O.~Edhan, R.~Jagadeesan, P.~Klemperer, and A.~Teytelboym (2020).
\newblock The equilibrium existence duality: Equilibrium with indivisibilities
  and income effects.
\newblock Working paper.

\bibitem[\protect\citeauthoryear{Baldwin, Jagadeesan, Klemperer, and
  Teytelboym}{Baldwin et~al.}{2021}]{baldwin2021consumer}
Baldwin, E., R.~Jagadeesan, P.~Klemperer, and A.~Teytelboym (2021).
\newblock On consumer theory with indivisible goods.
\newblock Working paper.

\bibitem[\protect\citeauthoryear{Baldwin, Jagadeesan, Klemperer, and
  Teytelboym}{Baldwin et~al.}{2023}]{baldwin2021equilibrium}
Baldwin, E., R.~Jagadeesan, P.~Klemperer, and A.~Teytelboym (2023).
\newblock The equilibrium existence duality.
\newblock {\em Journal of Political Economy\/}~{\em 131\/}(6), 1440--76.

\bibitem[\protect\citeauthoryear{Baldwin and Klemperer}{Baldwin and
  Klemperer}{2014}]{BaKl:14}
Baldwin, E. and P.~Klemperer (2014).
\newblock Tropical geometry to analyse demand.
\newblock Working paper.

\bibitem[\protect\citeauthoryear{Baldwin and Klemperer}{Baldwin and
  Klemperer}{2019}]{BaKl:19}
Baldwin, E. and P.~Klemperer (2019).
\newblock Understanding preferences: {``D}emand types,'' and the existence of
  equilibrium with indivisibilities.
\newblock {\em Econometrica\/}~{\em 87\/}(3), 867--932.

\bibitem[\protect\citeauthoryear{Baldwin, Klemperer, and Lock}{Baldwin
  et~al.}{2024}]{baldwin2024implementing}
Baldwin, E., P.~Klemperer, and E.~Lock (2024).
\newblock Implementing {W}alrasian equilibrium: The language of product-mix
  auctions.
\newblock Working paper.

\bibitem[\protect\citeauthoryear{Bikhchandani and Mamer}{Bikhchandani and
  Mamer}{1997}]{BiMa:97}
Bikhchandani, S. and J.~W. Mamer (1997).
\newblock Competitive equilibrium in an exchange economy with indivisibilities.
\newblock {\em Journal of Economic Theory\/}~{\em 74\/}(2), 385--413.

\bibitem[\protect\citeauthoryear{Candogan, Ozdaglar, and Parrilo}{Candogan
  et~al.}{2015}]{CaOzPa:15}
Candogan, O., A.~Ozdaglar, and P.~A. Parrilo (2015).
\newblock Iterative auction design for tree valuations.
\newblock {\em Operations Research\/}~{\em 63\/}(4), 751--71.

\bibitem[\protect\citeauthoryear{Danilov, Koshevoy, and Lang}{Danilov
  et~al.}{2003}]{DaKoLa:2003}
Danilov, V., G.~Koshevoy, and C.~Lang (2003).
\newblock Gross substitution, discrete convexity, and submodularity.
\newblock {\em Discrete Applied Mathematics\/}~{\em 131\/}(2), 283--98.

\bibitem[\protect\citeauthoryear{Danilov, Koshevoy, and Lang}{Danilov
  et~al.}{2013}]{danilov2013equilibria}
Danilov, V., G.~Koshevoy, and C.~Lang (2013).
\newblock Equilibria in markets with indivisible goods.
\newblock {\em Journal of the New Economic Association\/}~{\em 18\/}(2),
  10--34.

\bibitem[\protect\citeauthoryear{Danilov, Koshevoy, and Murota}{Danilov
  et~al.}{2001}]{DaKoMu:01}
Danilov, V., G.~Koshevoy, and K.~Murota (2001).
\newblock Discrete convexity and equilibria in economies with indivisible goods
  and money.
\newblock {\em Mathematical Social Sciences\/}~{\em 41\/}(3), 251--73.

\bibitem[\protect\citeauthoryear{Danilov and Koshevoy}{Danilov and
  Koshevoy}{2004}]{danilov2004discrete}
Danilov, V.~I. and G.~A. Koshevoy (2004).
\newblock Discrete convexity and unimodularity---{I}.
\newblock {\em Advances in Mathematics\/}~{\em 189\/}(2), 301--24.

\bibitem[\protect\citeauthoryear{Galeotti, Golub, Goyal, Talam{\`a}s, and
  Tamuz}{Galeotti et~al.}{2022}]{galeotti2022taxes}
Galeotti, A., B.~Golub, S.~Goyal, E.~Talam{\`a}s, and O.~Tamuz (2022).
\newblock Taxes and market power: A principal components approach.
\newblock Working paper.

\bibitem[\protect\citeauthoryear{Greenberg and Weber}{Greenberg and
  Weber}{1986}]{greenberg1986strong}
Greenberg, J. and S.~Weber (1986).
\newblock Strong {T}iebout equilibrium under restricted preferences domain.
\newblock {\em Journal of Economic Theory\/}~{\em 38\/}(1), 101--17.

\bibitem[\protect\citeauthoryear{Gul and Stacchetti}{Gul and
  Stacchetti}{1999}]{GuSt:99}
Gul, F. and E.~Stacchetti (1999).
\newblock Walrasian equilibrium with gross substitutes.
\newblock {\em Journal of Economic Theory\/}~{\em 87\/}(1), 95--124.

\bibitem[\protect\citeauthoryear{Gul and Stacchetti}{Gul and
  Stacchetti}{2000}]{GuSt:00}
Gul, F. and E.~Stacchetti (2000).
\newblock The {E}nglish auction with differentiated commodities.
\newblock {\em Journal of Economic Theory\/}~{\em 92\/}(1), 66--95.

\bibitem[\protect\citeauthoryear{Hatfield, Kominers, Nichifor, Ostrovsky, and
  Westkamp}{Hatfield et~al.}{2013}]{HaKoNiOsWe:11}
Hatfield, J.~W., S.~D. Kominers, A.~Nichifor, M.~Ostrovsky, and A.~Westkamp
  (2013).
\newblock Stability and competitive equilibrium in trading networks.
\newblock {\em Journal of Political Economy\/}~{\em 121\/}(5), 966--1005.

\bibitem[\protect\citeauthoryear{Hatfield, Kominers, Nichifor, Ostrovsky, and
  Westkamp}{Hatfield et~al.}{2019}]{HaKoNiOsWe:13}
Hatfield, J.~W., S.~D. Kominers, A.~Nichifor, M.~Ostrovsky, and A.~Westkamp
  (2019).
\newblock Full substitutability.
\newblock {\em Theoretical Economics\/}~{\em 14\/}(4), 1535--90.

\bibitem[\protect\citeauthoryear{Henry}{Henry}{1970}]{henry1970indivisibilites}
Henry, C. (1970).
\newblock Indivisibilit{\'e}s dans une {\'e}conomie d'{\'e}changes.
\newblock {\em Econometrica\/}~{\em 38\/}(3), 542--58.

\bibitem[\protect\citeauthoryear{Kelso and Crawford}{Kelso and
  Crawford}{1982}]{KeCr:82}
Kelso, A.~S. and V.~P. Crawford (1982).
\newblock Job matching, coalition formation, and gross substitutes.
\newblock {\em Econometrica\/}~{\em 50\/}(6), 1483--504.

\bibitem[\protect\citeauthoryear{Klemperer}{Klemperer}{2010}]{klemperer2010product}
Klemperer, P. (2010).
\newblock The product-mix auction: A new auction design for differentiated
  goods.
\newblock {\em Journal of the European Economic Association\/}~{\em 8\/}(2--3),
  526--36.

\bibitem[\protect\citeauthoryear{Ma}{Ma}{1998}]{Ma:98}
Ma, J. (1998).
\newblock Competitive equilibrium with indivisibilities.
\newblock {\em Journal of Economic Theory\/}~{\em 82\/}(2), 458--468.

\bibitem[\protect\citeauthoryear{McKenzie}{McKenzie}{1954}]{mckenzie1954equilibrium}
McKenzie, L. (1954).
\newblock On equilibrium in {G}raham's model of world trade and other
  competitive systems.
\newblock {\em Econometrica\/}~{\em 22\/}(2), 147--61.

\bibitem[\protect\citeauthoryear{Milgrom}{Milgrom}{2000}]{Milg:00}
Milgrom, P. (2000).
\newblock Putting auction theory to work: The simultaneous ascending auction.
\newblock {\em Journal of Political Economy\/}~{\em 108\/}(2), 245--272.

\bibitem[\protect\citeauthoryear{Milgrom}{Milgrom}{2009}]{milgrom2009assignment}
Milgrom, P. (2009).
\newblock Assignment messages and exchanges.
\newblock {\em American Economic Journal: Microeconomics\/}~{\em 1\/}(2),
  95--113.

\bibitem[\protect\citeauthoryear{Milgrom and Segal}{Milgrom and
  Segal}{2020}]{milgrom2020clock}
Milgrom, P. and I.~Segal (2020).
\newblock Clock auctions and radio spectrum reallocation.
\newblock {\em Journal of Political Economy\/}~{\em 128\/}(1), 1--31.

\bibitem[\protect\citeauthoryear{Milgrom and Strulovici}{Milgrom and
  Strulovici}{2009}]{MiSt:09}
Milgrom, P. and B.~Strulovici (2009).
\newblock Substitute goods, auctions, and equilibrium.
\newblock {\em Journal of Economic Theory\/}~{\em 144\/}(1), 212--47.

\bibitem[\protect\citeauthoryear{Nguyen and Teytelboym}{Nguyen and
  Teytelboym}{2024}]{nguyen2024equilibrium}
Nguyen, T. and A.~Teytelboym (2024).
\newblock Equilibrium in pseudomarkets.
\newblock Working paper.

\bibitem[\protect\citeauthoryear{Ostrovsky}{Ostrovsky}{2008}]{Ostr:08}
Ostrovsky, M. (2008).
\newblock Stability in supply chain networks.
\newblock {\em American Economic Review\/}~{\em 98\/}(3), 897--923.

\bibitem[\protect\citeauthoryear{Rostek and Yoder}{Rostek and
  Yoder}{2020}]{rostek2018matching}
Rostek, M. and N.~Yoder (2020).
\newblock Matching with complementary contracts.
\newblock {\em Econometrica\/}~{\em 88\/}(5), 1793--827.

\bibitem[\protect\citeauthoryear{Rostek and Yoder}{Rostek and
  Yoder}{2023}]{rostek2020complementarity}
Rostek, M. and N.~Yoder (2023).
\newblock Complementarity in matching, games, and exchange economies.
\newblock Working paper.

\bibitem[\protect\citeauthoryear{Rostek and Yoder}{Rostek and
  Yoder}{2024}]{rostek2023reallocative}
Rostek, M. and N.~Yoder (2024).
\newblock Reallocative auctions and core selection.
\newblock Working paper.

\bibitem[\protect\citeauthoryear{Schrijver}{Schrijver}{1998}]{schrijver1998theory}
Schrijver, A. (1998).
\newblock {\em Theory of Linear and Integer Programming}.
\newblock John Wiley \& Sons.

\bibitem[\protect\citeauthoryear{Shioura and Tamura}{Shioura and
  Tamura}{2015}]{shioura2015gross}
Shioura, A. and A.~Tamura (2015).
\newblock Gross substitutes condition and discrete concavity for multi-unit
  valuations: A survey.
\newblock {\em Journal of the Operations Research Society of Japan\/}~{\em
  58\/}(1), 61--103.

\bibitem[\protect\citeauthoryear{Sun and Yang}{Sun and Yang}{2006}]{SuYa:06}
Sun, N. and Z.~Yang (2006).
\newblock Equilibria and indivisibilities: Gross substitutes and complements.
\newblock {\em Econometrica\/}~{\em 74\/}(5), 1385--402.

\bibitem[\protect\citeauthoryear{Sun and Yang}{Sun and Yang}{2009}]{SuYa:09}
Sun, N. and Z.~Yang (2009).
\newblock A double-track adjustment process for discrete markets with
  substitutes and complements.
\newblock {\em Econometrica\/}~{\em 77\/}(3), 933--52.

\bibitem[\protect\citeauthoryear{Weinstein}{Weinstein}{2022}]{weinstein2022direct}
Weinstein, J. (2022).
\newblock Direct complementarity.
\newblock Working paper.

\bibitem[\protect\citeauthoryear{Yang}{Yang}{2017}]{yang2017maximal}
Yang, Y.-Y. (2017).
\newblock On the maximal domain theorem: A corrigendum to ``{W}alrasian
  equilibrium with gross substitutes''.
\newblock {\em Journal of Economic Theory\/}~{\em 172}, 505--11.

\end{thebibliography}

\clearpage
\appendix

\onehalfspacing
\section{Relationship to geometric approaches}
\label{app:connectionsToGeo}

In this section,
we explain the mathematical relationship between our equilibrium existence results
and the geometric approaches to equilibrium existence,
expanding on the introduction.
In Appendix~\ref{app:dkm},
we discuss the combinatorial geometric ``discrete convexity'' approach developed by \citet*{DaKoMu:01}.
In Appendix~\ref{app:bk19},
we discuss the tropical geometric approach developed by \citet*{BaKl:19}.

\subsection{Relationship to discrete convexity %
}
\label{app:dkm}

The key combinatorial geometric objects that \cite*{DaKoMu:01} worked with are classes of ``discrete convex sets.''
Such classes are closely related to unimodular sets of integer vectors \citep{danilov2004discrete}.
Throughout this appendix,
we let $\D \subseteq \mathbb{Z}^I$ consist of vectors whose components have no nontrivial common factors;
we also suppose that $\D$ is closed under negation.
Letting $\mathcal{D}$ be totally unimodular,
a finite set $S \subseteq \mathbb{Z}^I$ of integer vectors is \emph{$\mathcal{D}$-convex} if it includes all integer vectors in its convex hull,
and each edge of the convex hull is parallel to an element of $\mathcal{D}$.%
\footnote{For much of this section, we follow the terminology introduced by \cite*{danilov2013equilibria}.}

To develop their equilibrium existence result,
\citet*[page 261]{DaKoMu:01} assumed that $\zero \in X^j$ and studied the functions
\[q^j_m(\bun) = (\util{\cdot,\bun})^{-1}(\util{m,\zero})\]
They then introduced a notion of discrete convexity for functions:
letting $\D$ be a totally unimodular set of vectors,
a function $q^j: X^j \to \mathbb{R}$ is \emph{$\mathcal{D}$-convex} if for all $\p \in \mathbb{R}$,
the set
\[\argmin_{\bun \in X^j} \{q^j(\bun) + \p \cdot \bun\}\]
is $\mathcal{D}$-convex as a set of integer vectors.
Call a utility function $U^j$ is \emph{$\mathcal{D}$-quasiconcave} if for each real number $m$, the function $q^j_m$ is $\mathcal{D}$-convex.

    The classes of $\mathcal{D}$-quasiconcave utility functions for totally unimodular $\D$ are the \emph{classes of discrete convexity} introduced by \citet*{DaKoMu:01}.
Under some technical assumptions,
\citet*[Theorems 2 and 4]{DaKoMu:01} showed that competitive equilibria exist in exchange economies with indivisible goods if all agent's utility functions belong to the same class of discrete convexity.%
\footnote{%
The key additional technical condition \citet*{DaKoMu:01} imposed was an analogue of the standard free disposal condition.
Note that \citet*{DaKoMu:01} also allowed for production and for unbounded sets of feasible consumption vectors,
which are beyond the scope of this paper.
}

Mathematically,
there is a close connection between
our existence results and those of \citepossall{DaKoMu:01}.
The following proposition explains the connection:
any class of unit- and bundle-consistent preferences belongs to a class of discrete convexity.

\begin{prop}
\label{prop:dkmRelationshipConsistentToDc}
A family of utility functions belong to a single class of discrete convexity if and only if the family is unit- and bundle-consistent.
\end{prop}

In addition to illustrating the mathematical connection to \cite*{DaKoMu:01},
Proposition~\ref{prop:dkmRelationshipConsistentToDc} is technically useful in establishing Theorems~\ref{thm:existBinary} and~\ref{thm:existGen}.
More precisely,
we use the ``if'' direction of Proposition~\ref{prop:dkmRelationshipConsistentToDc} to prove Theorems~\ref{thm:existBinary} and~\ref{thm:existGen} by reducing these assertions to \citepossall{DaKoMu:01} combinatorial geometric equilibrium existence results
(and \citepossall{baldwin2020equilibrium} analogous tropical geometric equilibrium existence results).

Moreover, it turns out that classes of discrete convexity that can be characterized in terms of unit consistency and the compensated price effects:
the elements of $\mathcal{D}$ prescribe the possible compensated price effects that are in $\{-1,0,1\}^I$.

\begin{prop}
\label{prop:DquasiconcaveEquiv}
Let $\mathcal{D}$ be totally unimodular. %
A utility function is $\mathcal{D}$-quasiconcave if and~only if it is unit-consistent and each compensated price effect that lies in $\{-1,0,1\}^I$ lies in $\mathcal{D}$.
\end{prop}

We use the ``if'' direction of Proposition~\ref{prop:DquasiconcaveEquiv} to prove the ``if'' direction of Proposition~\ref{prop:dkmRelationshipConsistentToDc}.
The case of Proposition~\ref{prop:DquasiconcaveEquiv} under which agents demand at most one unit of each good and utility is quasilinear is a case of \citet*[Theorem 1]{baldwin2021consumer}.
The argument is much more involved beyond that case.
For the ``if'' direction, both the unit consistency of preferences and the total unimodularity of $\D$ play critical roles in ensuring that demand sets include all points in their convex hulls;
the formal proof relies on a mathematical result of \cite{danilov2004discrete}.
For the ``only if'' direction, price effects are not generally elements of $\{-1,0,1\}^I$, and deducing unit consistency relies on the total unimodularity of $\D$ and \citepossall{DaKoMu:01} existence results.

\subsection{Relationship to tropical geometric approaches%
}
\label{app:bk19}

In the quasilinear case,
\citet{BaKl:19} studied the set of price vectors at which demand is nonunique, which they call the ``locus of indifference prices'' (LIP) using methods from the mathematics of tropical geometry.
The ``tropical'' algebraic structure mathematically underlying quasilinear utility maximization implies that the LIP can be formally decomposed into linear components called ``facets.''
The geometric aspects of the LIP that \citet*{BaKl:19} focused on were the normal vectors to facets.

Formally, given a valuation $\valFn$,
\citet[Definition 2.2(2)]{BaKl:19} defined a \emph{LIP facet} to be a set $F \subseteq \mathbb{R}^I$ of price vectors,
that lies within exactly one hyperplane,
for which there are bundles $\bun,\bunpr \in X$ such that $F = \{\p \in \mathbb{R}^I \mid \bun,\bunpr \in \dQL{\p}\}$.
By construction,
each LIP facet has a well-defined normal direction,
which turns out to contain integer vectors \citep[page 873]{BaKl:19}.
Thus, %
\citet[Definition 3.1]{BaKl:19} defined a valuation $\valFn$ to be of \emph{demand type $\mathcal{D}$} if each LIP facet for $\valFn$ is normal to an element of $\mathcal{D}$.%
\footnote{\label{fn:falseClaimBK19}%
\citet[page 868]{BaKl:19} asserted, based on their Proposition 3.3,
that ``[demand type] vectors give the possible ways in which [...]~demand can change in response to a small generic price change.''
However,
small changes in prices can give changes in demand that are not in the directions of normal vectors to LIP facets even generically.
(Geometrically, this issue arises because generic small changes in prices can cross multiple LIP facets.)
For example, suppose that there are two goods and that the domain is $X^j = \{0,1\}^2$.
The additive valuation $\valFn$ defined by $\val{\bun} = x_1 + x_2$ is of demand type $\D = \{\pm(1,0),\pm(0,1)\}$.
Nevertheless, small changes in prices from $(1-\epsilon_1,1-\epsilon_2)$ to $(1+\delta_1,1+\delta_2)$ (for any $\epsilon_1,\epsilon_2,\delta_1,\delta_2 > 0$) change demand from $(1,1)$ to $(0,0)$---i.e., not by an element of $\D$.
The same issue applies to an analogous assertion of \citet*[page 24]{baldwin2020equilibrium} for settings with income effects.
}
Equivalently, and closer to \citet*{DaKoMu:01}, a valuation is of demand type $\D$ if for every price vector $\p$ for which $\Conv \dQL{\p}$ is one-dimensional,
$\Conv \dQL{\p}$ is parallel to an element of $\D$.%
\footnote{The equivalence follows from \citet[Proposition 2.20]{BaKl:19}; see also \citet*[Fact B.2]{baldwin2021consumer}.}
Calling a valuation \emph{pseudoconcave} if every demand set includes every integer point in its convex hull,
\citet[Theorem 4.3]{BaKl:19} showed that equilibria are guaranteed in transferable utility economies if all agents' valuations are pseudoconcave and of the same totally unimodular demand type.%
\footnote{\citet[Theorem 4.3]{BaKl:19} also
implies that if $\mathcal{D}$ contains the elementary basis vectors but is not totally unimodular,
then the class of pseudoconcave valuations of demand type $\mathcal{D}$ does not give domains for equilibrium existence.
}

\citeposs{BaKl:19} also apply slightly more generally when $\D$ is ``unimodular.''
However, the only unimodular demand types that contain valuations that are additive across goods are the totally unimodular demand types.
If a domain of preferences does not contain valuations that are additive across goods,
then agents cannot separately demand the goods,
which suggests that the natural objects that are bought are sold would be bundles rather than goods.
Redefining the natural bundles as goods converts any unimodular type to a totally unimodular one \citep[page 909]{BaKl:19}.
We therefore focus on total unimodularity,
rather than the slightly weaker property of unimodularity.

\citet*{baldwin2020equilibrium} extended some of \citeposs{BaKl:19} analysis to settings with income effects by reducing to the quasilinear case.
Given a utility function $\utilFn$ and a utility level $u$,
\citet*[Definition 1]{baldwin2021equilibrium} defined the \emph{Hicksian valuation} $V^j_{\mathrm{H}}: X^j \to \mathbb{R}$ by
\[\valH{\bun}{u} = - (U^j(\cdot,\bun))^{-1}(u);\]
their Lemma 1 (see Fact~\ref{fac:EEDlemma1}) shows that for all price vectors, demand for this Hicksian valuation matches Hicksian demand for the original utility function at utility level $u$.
\begin{fact}[\protect{\citealp*[Lemma 1]{baldwin2021equilibrium}}]
\label{fac:EEDlemma1}
For all price vectors $\p$ and utility levels $u$, we have that
\[\dH{\p}{u} = \argmax_{\bun \in X^j}\{\valH{\bun}{u} - \p \cdot \bun\}.\]
\end{fact}
\citet*[Definitions 9 and 10]{baldwin2020equilibrium} defined a utility function to be of \emph{demand type $\mathcal{D}$} if all of its Hicksian valuations are of demand type $\mathcal{D}$,
and to be \emph{quasipseudoconcave} if all of its Hicksian valuations are pseudoconcave.%
\footnote{\cite*{baldwin2020equilibrium} called this condition ``quasiconcavity'';
we use the term quasipseudoconcavity to parallel the term pseudoconcavity for a condition on valuations.}
\citet*[Theorem 3]{baldwin2020equilibrium} then showed that equilibria are guaranteed even in the presence of income effects if all agents' utility functions are quasipseudoconcave and of the same totally unimodular demand type.

As \citet*[Footnote 35]{baldwin2020equilibrium} showed,
if $\mathcal{D}$ is totally unimodular,
a utility function is quasipseudoconcave and of demand type $\mathcal{D}$ if and only if it is $\mathcal{D}$-quasiconcave.
Hence, totally unimodular demand types provide the same domains for equilibrium existence introduced by \cite*{DaKoMu:01}.
In particular,
Propositions~\ref{prop:dkmRelationshipConsistentToDc} and~\ref{prop:DquasiconcaveEquiv} therefore also give the mathematical relationship between unit and bundle consistency and the tropical geometric approach to equilibrium existence.

\begin{primeprop}{prop:dkmRelationshipConsistentToDc}
\label{prop:demTypeRelationshipConsistentToUnimod}
A family of valuations (resp.~utility functions) are all (quasi)pseudoconcave and belong to a single totally unimodular demand type if and only if the family is unit- and bundle-consistent.
\end{primeprop}

\begin{primeprop}{prop:DquasiconcaveEquiv}
\label{prop:demTypeEquivTotallyUnimod}
Let $\mathcal{D}$ be a totally unimodular. %
A valuation (resp.~utility function) is (quasi)pseudoconcave and of demand type $\mathcal{D}$ if and only if it is unit-consistent and each (compensated) price effect that lies in $\{-1,0,1\}^I$ lies in $\mathcal{D}$.
\end{primeprop}

Beyond the totally unimodular case,
it turns out that the characterization of the utility functions of a demand type in terms of (compensated) price effects given by Proposition~\ref{prop:demTypeEquivTotallyUnimod} persists provided that utility functions are assumed to be unit-consistent.

\begin{prop}
\label{prop:demTypeEquiv}
\begin{enumerate}[label=(\alph*)]
\item \label{part:demTypeOneDirectionNotConsistent}
If $\valFn$ (resp.~$\utilFn$) is unit-consistent and
each (compensated) price effect for agent $j$ that lies in $\{-1,0,1\}^I$ lies in $\mathcal{D}$,
then $\valFn$ (resp.~$\utilFn$) is of demand type $\D$.
\item \label{part:demTypeOneDirectionConsistentEquiv} If $\valFn$ (resp.~$\utilFn$) is unit-consistent and agent $j$ sees each pair of goods as either substitutes or complements,
then $\valFn$ (resp.~$\utilFn$) is of demand type $\D$ if and only if
each (compensated) price effect for agent $j$ that lies in $\{-1,0,1\}^I$ lies in $\mathcal{D}$.
\end{enumerate}
\end{prop}

In addition to further clarifying the mathematical relationship between our approach and the tropical geometric approach of \cite{BaKl:19} and \citet*{baldwin2020equilibrium},
Proposition~\ref{prop:demTypeEquiv} is technically useful in establishing %
the other propositions.

The cases of Propositions~\ref{prop:demTypeEquivTotallyUnimod} and~\ref{prop:demTypeEquiv}\ref{part:demTypeOneDirectionConsistentEquiv} under which agents demand at most one unit of each good and utility is quasilinear
corresponds to \citet*[Theorem 1]{baldwin2021consumer},
which does not require the hypothesis that each pair of goods is consistent for either direction.
The arguments are much more involved beyond that case,
where price effects are not generally elements of $\{-1,0,1\}^I$; and
critical roles are played by unit consistency, total unimodularity (for Proposition~\ref{prop:demTypeEquivTotallyUnimod}), and the consistency of each pair of goods (for one direction of Proposition~\ref{prop:demTypeEquiv}\ref{part:demTypeOneDirectionConsistentEquiv}).

\section{On the definitions of substitutes and complements}
\label{app:subsComps}

In this appendix, we prove that our definitions of substitutability and complementarity between \emph{pairs} of indivisible goods are equivalent to two other alternative definitions based on definitions from the literature that require that \emph{all} goods are substitutes.
We also connect our definitions to demand types.

Our definition of substitutability or complementarity for pairs of goods requires that for any decrease in the price $\pComp{i}$ of a good $i$ from prices $\p$ at which demand is unique,
there be a selection $\bunpr$ from demand at the new prices $(\pnewi)$ consistent with substitutability or complementarity.
By contrast,
in definitions of substitutability between all goods given in the literature,
\cite{KeCr:82} (see also \cite{MiSt:09}) allow demand to be non-unique at $\p$, and require that a selection $\bunpr$ from demand at $(\pnewi)$ exist for each selection $\bun$ from demand at $\p$.
On the other hand,
an analogous definition of \cite{AuMi:02} only considers new prices at which demand is unique.

In the case of the definition of substitutability between \emph{all} goods in which agents each demand at one unit of each indivisible good,
it is known that \citeposs{KeCr:82} and \citeposs{AuMi:02} definitions coincide with each other (\citealp*[Corollary 5] {DaKoLa:2003}; \citealp[Theorem A.1]{HaKoNiOsWe:13})
and hence with our definition.
However, the case of the definition of substitutability between \emph{all} goods when agents can demand more than one unit of an indivisible good,
\citeposs{KeCr:82} approach yields a strictly stronger condition than \citeposs{AuMi:02} definition  (\citealp*[Example 6]{DaKoLa:2003}; \citealp[Footnote 47]{BaKl:14}).

However, we can show that when in defining substitutability or complementarity between \emph{pairs} of indivisible goods,
extensions of \citeposs{KeCr:82} and \cite{AuMi:02} approaches coincide with each other,
and with our definition,
even when agents can demand more than one unit of any indivisible good.

\begin{lemma}
\label{lem:equivDefs}
The following properties are equivalent:
\begin{enumerate}[label=(\arabic*)]
\item \label{state:ourDef} Good $i$ is substitutable (resp.~complementary) to good $k$ for agent $j$.
\item \label{state:multi} For all utility levels $u$,
price vectors $\p$, new prices $\pprComp{i} < \pComp{i}$,
and $\bun \in \dH{\p}{u}$,
there exists $\bunpr \in \dH{\pnewi}{u}$ such that $\bunprComp{k} \le \bunComp{k}$
(resp.~$\bunprComp{k} \ge \bunComp{k}$).%
\footnote{This property is based on \citeposs{KeCr:82} definition of ``gross substitutes'' (see also \citeposs{MiSt:09} definition of ``weak substitutes'').
}
\item \label{state:unique} For all utility levels $u$,
price vectors $\p$, and new prices $\pprComp{i} < \pComp{i}$
with $\dH{\p}{u} = \{\bun\}$ and $\dH{\pnewi}{u} = \{\bunpr\}$,
we have that $\bunprComp{k} \le \bunComp{k}$
(resp.~$\bunprComp{k} \ge \bunComp{k}$).%
\footnote{This property is based on \citeposs{AuMi:02} definition of substitutes.
}
\item \label{state:demType} There exists $\D$ such that $\utilFn$ is of demand type $\D$ and for all $\dvec \in \D$, the product $\dvecComp{i}\dvecComp{k}$ is nonpositive (resp.~nonnegative).%
\footnote{The case of the equivalence between Properties~\ref{state:unique} and~\ref{state:demType} %
for which all goods are substitutes or all~goods are complements (rather than only a particular pair of goods), and $\utilFn$ is quasilinear, correspond to \citet[Propositions 3.6 and 3.8]{BaKl:19}.
See also \citet*[Theorem 1]{DaKoLa:2003}.
}
\end{enumerate}
\end{lemma}

\newcommand\eedlemma{Fact~\ref{fac:EEDlemma1}}
\newcommand\demTypeQuantity{Fact~\ref{fac:demandTypeQuantitySpace}}
\newcommand\demTypeQuantityIncEff{Fact~\ref{fac:EEDlemma1}}

\begin{proof}%
The implication \ref{state:multi}$\implies$\ref{state:unique} is obvious.
Hence, it suffices to prove the implications \ref{state:unique}$\implies$\ref{state:ourDef}$\implies$\ref{state:multi}
and \ref{state:ourDef}$\implies$\ref{state:demType}$\implies$\ref{state:unique}.
To prove these implications,
in light of \eedlemma,
we can assume that $\utilFn$ is quasilinear.

\newcommand\svec{\mathbf{s}}

To prove the implication \ref{state:unique}$\implies$\ref{state:ourDef},
let $\p$ be a price vector with $\dQL{\p} = \{\bun\}$, and let $\pprComp{i} < \pComp{i}$ be a new price.
By the upper hemicontinuity of demand,
there exists $\varepsilon$ such that $\dQL{\p + \svec} \subseteq \dQL{\p}$ for all $\|\svec\| < \varepsilon$.
Writing $\ppr = (\pnewi)$ and letting $\bunpr$ be an extreme point of $\Conv \dQL{\ppr}$,
by \citet*[Claim B.2]{baldwin2021consumer},
there exists $\svec \in \mathbb{R}^I$ with $\|\svec\| < \varepsilon$ such that $\dQL{\ppr + \svec} = \{\bunpr\}$.
Applying Property~\ref{state:unique} to the decrease in the price of good $i$ from $\pComp{i} + s_i$ to $\pprComp{i} + s_i$ starting at price vector $\p + \svec,$
noting that $\dQL{\p + \svec} = \{\bun\}$ by construction,
we have that $\hbunComp{k} \le \bunComp{k}$ (resp.~$\hbunComp{k} \ge \bunComp{k}$)---yielding \ref{state:ourDef}.

To prove the implication \ref{state:ourDef}$\implies$\ref{state:multi},
let $\p$ be a price vector, let $\pprComp{i} < \pComp{i}$ be a new price,
and let $\bun \in \dQL{\p}$.
Let $\hbun$ be an extreme point of $\Conv \dQL{\p}$ with lowest (resp.~highest) $k$-component.
Writing $\ppr = (\pnewi)$,
by the upper hemicontinuity of demand,
there exists $\varepsilon$ such that $\dQL{\ppr + \svec} \subseteq \dQL{\ppr}$ for all $\|\svec\| < \varepsilon$.
By \citet*[Claim B.2]{baldwin2021consumer},
there exists $\svec \in \mathbb{R}^I$ with $\|\svec\| < \varepsilon$ such that $\dQL{\p + \svec} = \{\hbun\}$.
Applying Property~\ref{state:ourDef} to the decrease in the price of good $i$ from $\pComp{i} + s_i$ to $\pprComp{i} + s_i$ starting at price vector $\p + \svec,$
there exists $\bunpr \in \dQL{\ppr + \svec}$ with $\bunprComp{k} \le \hbunComp{k}$ (resp.~$\bunprComp{k} \ge \hbunComp{k}$).
By the construction of $\hbun$ and the Krein-Millman Theorem,
we have that $\hbunComp{k} \le \bunComp{k}$ (resp.~$\hbunComp{k} \ge \bunComp{k}$),
so $\bunprComp{k} \le \bunComp{k}$ (resp.~$\bunprComp{k} \ge \bunComp{k}$).
The definition of $\varepsilon$,
implies that $\bunpr \in \dQL{\ppr}$---yielding \ref{state:multi}.

To prove the implication \ref{state:ourDef}$\implies$\ref{state:demType},
we prove the contrapositive.
Suppose that for all $\D$ such that $\valFn$ is of demand type $\D$, there exists $\dvec \in \D$ such the product $\dvecComp{i}\dvecComp{k}$ is negative (resp.~positive).
Then, %
there exists a price vector $\p$ such that $\Conv \dQL{\p}$ is a line segment parallel to a vector $\dvec$ such the product $\dvecComp{i}\dvecComp{k}$ is positive (resp.~negative).
Without loss of generality, assume that $\dvecComp{i} > 0$, so $\dvecComp{k} > 0$ (resp.~$\dvecComp{k} < 0$).
Let the endpoints of this line segment be $\bun$ and $\bunpr$,
where $\bunpr - \bun = m \dvec$ and $m > 0$.
\citet*[Claim B.1]{baldwin2021consumer} implies that there exists $\varepsilon > 0$ such that $\dQL{\p + \varepsilon \e{i}} = \{\bun\}$.
Thus, considering the reduction in the price of $i$ by $\varepsilon$ starting from price vector $\p + \varepsilon \e{i}$,
we see that $i$ is not substitutable (resp.~complementary) to $k$.

The quasilinear case of the implication \ref{state:demType}$\implies$\ref{state:unique} follows from \citet*[Proposition D.2]{baldwin2021consumer}.
\end{proof}

\section{Proofs}
\label{app:proofs}

We present proofs of results in logical order, which differs from the order in which results are stated in the text.
The logical dependencies between results are as follows.
\vspace{-18pt}
\tikzcdset{row sep/normal=10pt}
\tikzcdset{column sep/normal=12pt}
\usetikzlibrary{decorations.pathmorphing}
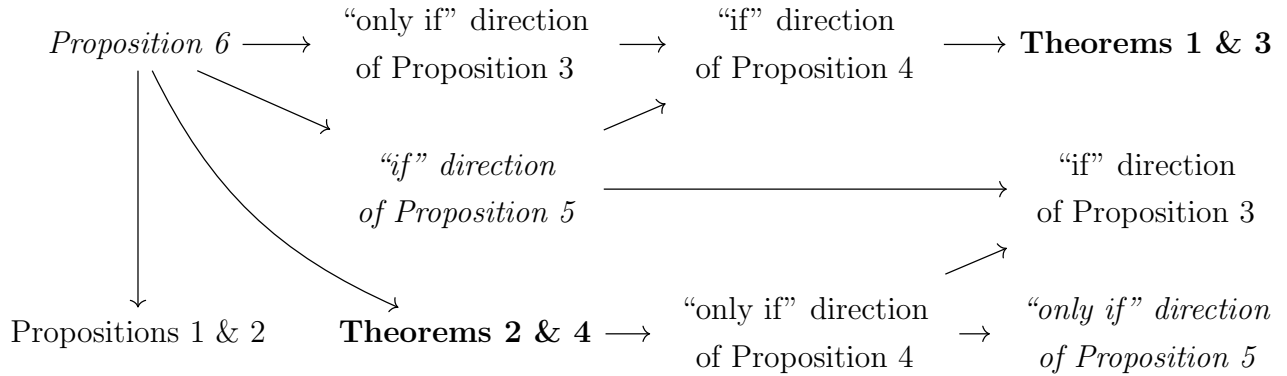
\[
\hspace{-10pt}
\begin{tikzcd}
\text{\emph{Proposition~\ref{prop:demTypeEquiv}}} \ar{r} \ar{dd} \ar{rd} \ar[bend right=20]{rdd} & \begin{tabular}{c}
\text{``only if'' direction}\\
\text{of Proposition~\ref{prop:totallyunimodgen}}
\end{tabular} \ar{r}
& \begin{tabular}{c}
\text{``if'' direction}\\
\text{of Proposition~\ref{prop:dkmRelationshipConsistentToDc}}
\end{tabular} \ar{r} & \text{\textbf{Theorems~\ref{thm:existBinary} \&~\ref{thm:existGen}}}\\
& \begin{tabular}{c}
\text{\emph{``if'' direction}}\\
\text{\emph{of Proposition~\ref{prop:DquasiconcaveEquiv}}}
\end{tabular} \ar{ru} \ar{rr} && \begin{tabular}{c}
\text{``if'' direction}\\
\text{of Proposition~\ref{prop:totallyunimodgen}}
\end{tabular} \\
\text{Propositions~\ref{prop:transitivityCompBinary} \&~\ref{prop:transitivityCompGen}} & \text{\textbf{Theorems~\ref{thm:neccBinary} \&~\ref{thm:neccGen}}} \ar{r} & \begin{tabular}{c}
\text{``only if'' direction}\\
\text{of Proposition~\ref{prop:dkmRelationshipConsistentToDc}}
\end{tabular} \ar{r} \ar{ru} & \begin{tabular}{c}
\text{\emph{``only if'' direction}}\\
\text{\emph{of Proposition~\ref{prop:DquasiconcaveEquiv}}}
\end{tabular}
\end{tikzcd}
\]
Our main results are in bold.
Results in italics are not needed in the case of one unit of each good (where \citet*[Theorem 1]{baldwin2021consumer} applies);
Footnotes~\ref{fn:binaryTransitivityComp},~\ref{fn:binaryNonUnimodToInconsistent},~\ref{fn:dkmRelationship},  and~\ref{fn:binaryNecc} discuss a shorter line of argument that is applicable (only) in that case.
Note also that
Lemma~\ref{lem:compPriceEffConvexComb}, stated in the proof of Proposition~\ref{prop:demTypeEquiv}\ref{part:demTypeOneDirectionConsistentEquiv}, is also used in the proof of Propositions~\ref{prop:transitivityCompBinary} and~\ref{prop:transitivityCompGen};
and
Lemma~\ref{lem:nonUnimodToInconsistent}, stated in the proof of the ``only if'' direction of Proposition~\ref{prop:totallyunimodgen}, is also used in the proof of Theorems~\ref{thm:neccBinary} and~\ref{thm:neccGen}.

\subsection{Preliminaries}

We first state a lemma and recall a fact from the literature that we use in several proofs.

The following lemma, which we use in the proofs of Proposition~\ref{prop:demTypeEquiv}\ref{part:demTypeOneDirectionNotConsistent} and the ``if'' directions of Propositions~\ref{prop:DquasiconcaveEquiv} and~\ref{prop:demTypeEquivTotallyUnimod},
shows that unit consistency implies the ``consecutive integer property'' introduced by \citet[Theorem 9]{MiSt:09}.

\begin{lemma}
\label{lem:unitConsistentConsecutiveInteger}
If $\utilFn$ is unit-consistent, then for each price vector $\p$, utility level $u$, and good $i$,
the set $\{\bunComp{i} \mid \bun \in \dH{\p}{u}\}$ consists of consecutive integers.%
\end{lemma}

\citet[Theorem 9]{MiSt:09} showed the conclusion of Lemma~\ref{lem:unitConsistentConsecutiveInteger} under the hypothesis that an agent with a quasilinear utility function see all units of all goods as substitutes.
We extend this conclusion to allow agents to see units of different goods as complements (and allow to allow for income effects).

\begin{proof}
\newcommand\hvalFn{\widehat{V}^j}
\newcommand\hval[1]{\hvalFn(#1)}
\newcommand\ohdQLFn{\overline{\widehat{D}}^j}
\newcommand\ohdQL[1]{\ohdQLFn(#1)}
By Fact~\ref{fac:EEDlemma1},
we can assume that agent $j$'s utility function is quasilinear.
Let $\p$ be a price vector, and let $i$ be a good.
Let $S = \{\bunComp{i} \mid \bun \in X^j\},$
and consider the single-good valuation $\hvalFn: S \to \mathbb{R}$ defined by
\[\hval{\bunComp{i}} = \max_{\bun_{I \ssm \{i\}} \mid (\bunComp{i},\bun_{I \ssm \{i\}}) \in X^j} \{\val{\bunComp{i},\bun_{I \ssm \{i\}}} - \p_{I \ssm \{i\}} \cdot \bun_{I \ssm \{i\}}\}.\]
Note that demand under $\hvalFn$ at price $\pComp{i}$ is $\{\bunComp{i} \mid \bun \in \dH{\p}{u}\}$ by construction.
Hence, it suffices to prove that demand under $\hvalFn$ at each price consists of a set of consecutive integers.

Writing $\odQL{\overline{\p}_{i \times \{1,2,\ldots,M\}}$ for the demand under $\hvalFn$ when each unit of $i$ is regarded as a separate item, and items are priced according to $\overline{p}_{i \times \{1,2,\ldots,M\}}} \in \mathbb{R}^{i \times \{1,2,\ldots,M\}}$,
we have that
\[\ohdQL{\overline{\p}_{i \times \{1,2,\ldots,M\}}} = \big\{\overline{\bun}_{i \times \{1,2,\ldots,M\}} \big|  \overline{\bun} \in \odQL{\overline{\p}_{i \times \{1,2,\ldots,M\}},\overline{\p}_{(I \ssm \{i\}) \times \{1,2,\ldots,M\}}}\big\},\]
where we define $\overline{\p}_{(I \ssm \{i\}) \times \{1,2,\ldots,M\}} \in \mathbb{R}^{(I \ssm \{i\}) \times \{1,2,\ldots,M\}}$ by $(\overline{\p}_{(I \ssm \{i\}) \times \{1,2,\ldots,M\}})_{(k,m)} = \pComp{k}$.
By the \ref{state:ourDef}$\implies$\ref{state:multi} implication of Lemma~\ref{lem:equivDefs},
it follows that when all units of goods are regarded as separate items, items are substitutes under $\hvalFn$.
\citet[Theorem 9]{MiSt:09} hence implies that demand under $\hvalFn$ at each price consists of a set of consecutive integers.
\end{proof}

The proofs of the ``if'' directions of Propositions~\ref{prop:DquasiconcaveEquiv} and~\ref{prop:demTypeEquivTotallyUnimod},
and of Theorems~\ref{thm:neccBinary} and~\ref{thm:neccGen},
use the following fact regarding equilibrium prices in transferable utility economies.

\begin{fact}[\protect{\citealp[Proposition 2 and Theorem 18]{MiSt:09}; \citealp[Lemma 2.11]{BaKl:19}}]
\label{fac:pseudoequilPrices}
In transferable utility economies,
if competitive equilibria exist and a price vector $\p$ satisfies
\[\sum_{j \in J} \e{j} \in \Conv \sum_{j \in J} \dQL{\p},\]
then $\p$ is a price vector of a competitive equilibrium.
\end{fact}

\subsection{Proof of Proposition~\ref{prop:demTypeEquiv}}

\paragraph*{Proof of Part~\ref{part:demTypeOneDirectionNotConsistent}.}

By Fact~\ref{fac:EEDlemma1},
we can assume that agent $j$'s utility function is quasilinear.
Consider any price vector $\p$ and utility level $u$ such that $\Conv \dQL{\p}$ is one-dimensional.
Let $\dvec$ be an integer vector with no non-trivial common factors such that $\Conv \dQL{\p}$ is parallel to $\dvec$,
and let $i \in \argmax_{i' \in I} |\dvecComp{i'}|$.
Let $\bun,\bunpr \in \dQL{\p}$ be the elements of $\dQL{\p}$ with the lowest and second-lowest $i$-components, respectively.

Lemma~\ref{lem:unitConsistentConsecutiveInteger}
implies that $\bunprComp{i} - \bunComp{i} = 1$.
As $\bunpr - \bun$ must be an integer multiple of $\dvec$,
we must have that $\bunpr - \bun = \pm \dvec$.
In particular, we have that $\dvecComp{i} = \pm 1$;
by the definition of $i$, it follows that $\dvec \in \{-1,0,1\}^I$,
and hence that $\bunpr - \bun \in \{-1,0,1\}^I$.

By Fact~\ref{fac:EEDlemma1} and \citet*[Claim B.1]{baldwin2021consumer},
there exists $\hpComp{i} > \pComp{i}$ such that $\dQL{\hpComp{i},\p_{I \ssm \{i\}}} = \{\bun\}$.
Considering the decrease in the price of $i$ from $\hpComp{i}$ to $\pComp{i}$,
we see that $\bunpr - \bun$ is a compensated price effect for agent $j$ and good $i$.

Since $\p$ was arbitrary with $\Conv \dQL{\p}$ one-dimensional,
$\valFn$ is of demand type $\D$.

\paragraph{Proof of Part~\ref{part:demTypeOneDirectionConsistentEquiv}.}

Part~\ref{part:demTypeOneDirectionNotConsistent} provides the ``if'' direction.
The proof of the ``only if'' direction relies on the following lemma,
which we also use in the proof of Propositions~\ref{prop:transitivityCompBinary} and~\ref{prop:transitivityCompGen}.

\begin{lemma}
\label{lem:compPriceEffConvexComb}
Suppose that agent $j$'s preferences are unit-consistent and of demand type $\D$.
If $\Delta \bun \in \{-1,0,1\}^I$ is a compensated price effect for agent $j$ and a good $i$,
then $\Delta \bun$ is a convex combination of elements of
\[\D_i = \{\dvec \in \D \cap \{-1,0,1\}^I \mid \dvecComp{i} = 1\}.\]
\end{lemma}
\begin{proof}
By Fact~\ref{fac:EEDlemma1},
we can assume that agent $j$'s utility function is quasilinear.
Without loss of generality, we can take $\D$ to be minimal such that $\valFn$ is of demand type $\D$.
It then follows from Proposition~\ref{prop:demTypeEquiv}\ref{part:demTypeOneDirectionNotConsistent} that $\D \subseteq \{-1,0,1\}^I$.
Note that hence
$\D_i = \{\dvec \in \D \mid \dvecComp{i} > 0\}.$

Let $\p$ be a price vector and let $\pprComp{i} < \pComp{i}$ be a new price for $i$
with $\dQL{\p} = \{\bun\}$ such that $\bun + \Delta \bun \in \dQL{\pnewi}$.
Applying \citet*[Proposition 2]{baldwin2021consumer} to the increase in the price of good $i$ from $\pprComp{i}$ to $\pComp{i}$ with demand starting at the selection $\bun + \Delta \bun$,
we see that $\Delta \bun$ must be a nonnegative linear combination of elements of $\D_i$.
In particular, 
we have that $(\Delta \bun)_i > 0$, hence that $(\Delta \bun)_i = 1.$
Writing
\[\Delta \bun = \sum_{\dvec \in \D_i} \alpha_{\dvec} \dvec,\]
with $\alpha_{\dvec} \ge 0$ for all $\dvec \in \D_i,$
since $\dvecComp{i} = 1$ for all $\dvec \in \D_i$, we have that $\sum_{\dvec \in \D_i} \alpha_{\dvec} = 1$.
\end{proof}

To complete the proof of Part~\ref{part:demTypeOneDirectionConsistentEquiv} of the proposition,
let $\Delta \bun \in \{-1,0,1\}^I$ be a compensated price effect for agent $j$ and a good $i$;
we show that $\Delta \bun \in \D$.
By Lemma~\ref{lem:compPriceEffConvexComb}, we can write
$\Delta \bun$ as a convex combination of elements of as $\D_i$,
say as
\[\Delta \bun = \sum_{\dvec \in \D_i} \alpha_{\dvec} \dvec.\]

We claim that there exists $\dvec \in \D_i$ such that $\alpha_{\dvec} = 1$.
Suppose for sake of deriving a contradiction that no such $\dvec \in \D_i$ exists.
Then, there exist distinct $\dvec',\dvec'' \in \D_i$ with $\alpha_{\dvec'},\alpha_{\dvec''} > 0$.
Let $k$ be a good such that $\dvecprComp{k} \not= \dvecComp{k}''$.
By the \ref{state:ourDef}$\implies$\ref{state:demType} implication of Lemma~\ref{lem:equivDefs}, we must have that $\dvecprComp{k},\dvecComp{k}'' \ge 0$ or that $\dvecprComp{k},\dvecComp{k}'' \le 0$.
Without loss of generality, assume that $\dvecprComp{k} = \pm 1$ and that $\dvecComp{k}' = 0$.
The \ref{state:ourDef}$\implies$\ref{state:demType} implication of Lemma~\ref{lem:equivDefs} then guarantees that $\dvecComp{k} \in \{0,\dvecprComp{k}\}$ for all $\dvec \in \D_i$.
It follows that
\[|(\Delta \bun)_k| = \sum_{\dvec \in \D_i \mid \dvecComp{k} = \dvecprComp{k}} \alpha_{\dvec}.\]
This quantity is strictly between 0 and 1 since $\sum_{\dvec \in \D_i} \alpha_{\dvec} = 1$ and $\alpha_{\dvec'},\alpha_{\dvec''} > 0$%
---contradicting the hypothesis that $\Delta \bun \in \mathbb{Z}^I$.

Hence, we can conclude there exists $\dvec \in \D_i$ such that $\alpha_{\dvec} = 1$.
It follows that $\Delta \bun = \dvec \in \D$.

\subsection{Proof of Propositions~\ref{prop:transitivityCompBinary} and~\ref{prop:transitivityCompGen}}

Proposition~\ref{prop:transitivityCompBinary} is a special case of Proposition~\ref{prop:transitivityCompGen} because if an agent demands at most one unit of each good,
their preferences are unit-consistent and all compensated price effects are elements of $\{-1,0,1\}^I$.%
\footnote{\label{fn:binaryTransitivityComp}Proposition~\ref{prop:transitivityCompBinary} alone can also be obtained by combining \citet*[Theorem 1]{baldwin2021consumer} with Lemma~\ref{lem:equivDefs}.  This argument, however, does not extend to give a proof of Proposition~\ref{prop:transitivityCompGen}.}
Hence, it suffices to prove Proposition~\ref{prop:transitivityCompGen}.

To prove Proposition~\ref{prop:transitivityCompGen},
suppose that $\utilFn$ is of demand type $\D$.
By Lemma~\ref{lem:compPriceEffConvexComb},
we can write
$\Delta \bun$ as a convex combination of elements of $\D_i$, say as
\[\Delta \bun = \sum_{\dvec \in \D_i} \alpha_{\dvec} \dvec.\]

\paragraph*{Proof of Part~\ref{part:transitivityCompGenComp}.}

Since $\Delta \bun \in \{-1,0,1\}^I$,
we have that $(\Delta \bun)_k = (\Delta \bun)_\ell = 1$ (resp.~$(\Delta \bun)_k = (\Delta \bun)_\ell = -1$).
Since $\D_i \subseteq \{-1,0,1\}^I$,
we must have that $\dvecComp{k} = \dvecComp{\ell} = 1$ (resp.~$\dvecComp{k} = \dvecComp{\ell} = -1$) for all $\dvec \in \D_i$ with $\alpha_{\dvec} > 0$.
In particular, there exists $\dvec \in \D$ with $\dvecComp{k}\dvecComp{\ell}>0$.
Since $\D$ was arbitrary with $\utilFn$ of demand type $\D$,
the contrapositive of the \ref{state:unique}$\implies$\ref{state:demType} implication of Lemma~\ref{lem:equivDefs} implies that $k$ and $\ell$ must be strict complements at some prices.

\paragraph*{Proof of Part~\ref{part:transitivityCompGenSubs}.}

Since $\Delta \bun \in \{-1,0,1\}^I$,
we have that $(\Delta \bun)_k = 1$ and that $(\Delta \bun)_\ell = -1$.
Since $\D_i \subseteq \{-1,0,1\}^I$,
we must have that $\dvecComp{k} = 1$ and $\dvecComp{\ell} = -1$ for all $\dvec \in \D_i$ with $\alpha_{\dvec} > 0$.
In particular, there exists $\dvec \in \D$ with $\dvecComp{k}\dvecComp{\ell}<0$.
Since $\D$ was arbitrary with $\utilFn$ of demand type $\D$,
the contrapositive of the \ref{state:unique}$\implies$\ref{state:demType} implication of Lemma~\ref{lem:equivDefs} implies that there must be a strict hidden complementarity between $k$ and $\ell$ at some prices.

\subsection{Proof of the ``only if'' direction of Proposition~\ref{prop:totallyunimodgen}}

The proof relies on a more technical connection between consistent bundles and total unimodularity,
which we also use in the proof of Theorems~\ref{thm:neccBinary} and~\ref{thm:neccGen}.

\newcommand\svec{\mathbf{s}}

\begin{lemma}
\label{lem:nonUnimodToInconsistent}
Suppose that preferences are unit-consistent and that each agent sees each pair of goods as either complements or substitutes.
Let $S \subseteq \{-1,0,1\}^I$ be a set of compensated price effects.
If
$S$ is not totally unimodular,
then there exists a bundling $\B \subseteq S \cup \{\e{i} | i \in I\}$ that contains a pair of inconsistent bundles but excludes at least two elements of $S$.
\end{lemma}

To prove Lemma~\ref{lem:nonUnimodToInconsistent}, we use the following technical linear algebraic claim.

\begin{claim}
\label{cl:soupedUpSchrijver}
Suppose that $S \subseteq \{-1,0,1\}^I$ is not totally unimodular.
Then, there exist linearly independent vectors $\svec^1,\ldots,\svec^{|I|} \in S \cup \{\e{i} \mid i \in I\}$ and vectors $\dvec,\dvec' \in S \ssm \{\svec^1,\ldots,\svec^{|I|}\}$ such that letting $G$ be the $I \times |I|$ matrix whose columns are $\svec^1,\ldots,\svec^{|I|},$
there exist $1 \le k,\ell \le |I|$ with $(G^{-1} \dvec)_k (G^{-1} \dvec)_\ell > 0$ and $(G^{-1} \dvec')_k (G^{-1} \dvec')_\ell < 0$.
\end{claim}
\begin{proof}
By removing vectors from $S$ if necessary,
we can assume that $S$ does not contain any elementary basis vectors or negations thereof.
By further removing vectors from $S$ if necessary,
we can assume that there exists a set $I' \subseteq I$ of goods such that the matrix with columns $S \cup \{\e{i} | i \in I \ssm I'\}$ is square of determininant not $0$ or $\pm 1$.

The remainder of the proof follows \citet[pages 270--271]{schrijver1998theory}.
Consider a square matrix $B$ whose columns are elements of $\{\svec_{I'} \mid \svec \in S\}$.
By construction, we have $|\det B| > 1.$
Consider the matrix $C = [B \,\, \mathrm{Id}_{|I'| \times |I'|}]$;
and a matrix $C'$ obtained from $C$ by adding and subtracting rows and multiplying rows by $-1$
such that (i) all entries in the first $|I|$ columns of $C'$ are $0$ or $\pm 1$,
(ii) the columns of $C'$ include the elementary basis vectors $\e{i}$ for $i \in I''$,
and (iii) the first $|I'|$ columns of $C'$ include as many elementary basis vectors as possible.

Since $|\det B| > 1,$
not all of the elementary basis vectors $\e{i}$ for $i \in I''$ can be among the first $|I'|$ columns of $C'$.
Consider a column $\cvec$ among the first $|I'|$ columns of $C'$ that is not an elementary basis vectors.
By construction, the column must contain a nonzero entry;
without loss of generality, suppose that it is in row $k$.
By (iii), it must be impossible to transform the column to $\e{k}$ by adding or subtracting row $k$ to or from other rows.
By (i), it follows that there is a column $\cvec'$ among the first $|I|$ columns and a row $\ell$ such that the $k$th and $\ell$th rows of columns $\cvec$ and $\cvec'$ form a submatrix with three entries of $1$ and one entry of $-1$, or vice versa.
Without loss of generality, assume that $\cvecComp{k}\cvecComp{\ell} > 0 > \cvecComp{k}'\cvecComp{\ell}'$.

For each $1 \le n \le |I''|$, suppose that the $\sigma(n)$th column of $C'$ is the $n$th elementary basis vector,
and let $\svec^n$ be the $\sigma(n)$th column of $C$.
Let $\svec^{|I'|+1},\ldots,\svec^{|I|}$ be the elementary basis vectors $\{\e{i} \mid i \in I \ssm I'\}$ in any order.
Similarly,
let $\dvec$ and $\dvecpr$ be the columns of $C$ corresponding to the columns $\cvec$ and $\cvec'$ of $C'$, respectively.
By construction, the first $|I'|$ components of $G^{-1} \dvec$ (resp.~$G^{-1} \dvecpr$) are given by $\cvec$ (resp.~$\cvec'$), and the lemma follows.
\end{proof}

\begin{proof}[Proof of Lemma~\ref{lem:nonUnimodToInconsistent}]
Let $S' \subseteq S$ be a minimal subset that is not totally unimodular.
Claim~\ref{cl:soupedUpSchrijver} implies that
there exist linearly independent vectors $\svec^1,\ldots,\svec^{|I|} \in S' \cup \{\e{i} \mid i \in I\}$,
vectors $\dvec,\dvec' \in S' \ssm \{\svec^1,\ldots,\svec^{|I|}\}$, and $1 \le k,\ell \le |I|$
such that letting $G$ be the $I \times |I|$ matrix whose columns are $\svec^1,\ldots,\svec^{|I|},$
we have that $(G^{-1} \dvec)_k (G^{-1} \dvec)_\ell > 0$ and that $(G^{-1} \dvec')_k (G^{-1} \dvec')_\ell < 0$.
The minimality of $S'$ implies that $G$ must have determinant $\pm 1$.

\newcommand\tutilFn{\widetilde{U}^j}
\newcommand\tutil[1]{\tutilFn(#1)}

For each agent $j$, define a transformed utility function $\tutilFn: \mathbb{R} \times (G^{-1} X^j) \to \mathbb{R}$ by $\tutil{x_0,\q} = \util{x_0,G\q}$.%
\footnote{This extends a construction of \citet*[page 885]{BaKl:19} to settings with income effects.}
By Proposition~\ref{prop:demTypeEquiv}\ref{part:demTypeOneDirectionConsistentEquiv},
all agents' preferences are of demand type $\D$ only if $\D \supseteq S$.
\citet[Proposition 3.11]{BaKl:19} then implies that all the transformed preferences are of demand type $\mathcal{D}$ only if $\mathcal{D} \supseteq G^{-1} S$.%
\footnote{\label{fn:binaryNonUnimodToInconsistent}The case of Lemma~\ref{lem:nonUnimodToInconsistent} for which each agent demands at most one unit of each good (and hence the same case of Proposition~\ref{prop:totallyunimodgen}) can be obtained without the use of Proposition~\ref{prop:demTypeEquiv} here by instead using \citet*[Theorem 1]{baldwin2021consumer}.  This argument, however, does not extend to the general case in which agents can demand more than one unit of each good.}
Since $(G^{-1} \dvec)_k (G^{-1} \dvec)_\ell > 0$,
the contrapositive of the \ref{state:ourDef}$\implies$\ref{state:demType} implication of Lemma~\ref{lem:equivDefs} implies that under the transformed utility functions,
the $k$th and $\ell$th commodities are not complements.
Similarly, since $(G^{-1} \dvec')_k (G^{-1} \dvec')_\ell < 0$,
under the transformed utility functions,
the $k$th and $\ell$th commodities are not substitutes.

Consider the bundling $\B = \{\svec^1,\ldots,\svec^{|I|}\}$.
The definition of the transformed utility functions implies that $\dHB{\tp}{u}{\B}$ is Hicksian demand for $\tutilFn$ at price vector $\tp$ and utility level $u$.
Hence, for the bundling $\B$, the bundles $\svec^k$ and $\svec^\ell$ are neither complements nor substitutes.
\end{proof}

To complete the proof of the ``only if'' directions of Proposition~\ref{prop:totallyunimodgen},
note bundle consistency implies that each agent sees each pair of goods as either complements or substitutes.
Let $S$ be the set of all compensated price effects for all agents and goods that lie in $\{-1,0,1\}^I$,
and let $\D = S \cup -S$.
As any bundling $\B \subseteq S \cup \{\e{i} \mid i \in I\}$ consists of relevant bundles,
it follows from Lemma~\ref{lem:nonUnimodToInconsistent} that $\D$ is totally unimodular.
Hence, $S$ must be totally unimodular~as~well.

\subsection{Proof of the ``if'' directions of Propositions~\ref{prop:DquasiconcaveEquiv} and~\texorpdfstring{\ref{prop:demTypeEquivTotallyUnimod}}{\ref{prop:DquasiconcaveEquiv}'}}

Consider a unit-consistent utility function $\utilFn$ all of whose compensated price effects that lie in $\{-1,0,1\}^I$ also lie in $\mathcal{D}$.
Proposition~\ref{prop:demTypeEquiv}\ref{part:demTypeOneDirectionNotConsistent} guarantees that $\utilFn$ is of demand type $\D$.
It remains to prove that $\utilFn$ is quasipseudoconcave.
By adding elementary basis vectors to $\D$ if necessary, %
we can assume that $\D$ includes all the elementary basis vectors.

Fix a utility level $u$;
we show that the Hicksian valuation $\valH{\cdot}{u}$ is pseudoconcave.
\eedlemma{} implies that demand under the valuation $\valH{\cdot}{u}$ is $\dH{\cdot}{u}$.
Hence, it suffices to show that for all price vectors $\p$, the set $\dH{\p}{u}$ includes all integer vectors in its convex~hull.

We prove this assertion by induction on the dimension of $\dH{\p}{u}$.
The base case, in which $\dH{\p}{u}$ has dimension 0, is obvious as $|\dH{\p}{u}| = 1$ in that case.

Suppose that the assertion holds for all price vectors $\ppr$ for which $\Conv \dH{\ppr}{u}$ has dimension $n$.
Let $\p$ be a price vector such that $\Conv \dH{\p}{u}$ has dimension $n + 1$.

To show that $\dH{\p}{u}$ contains all integer vectors in its convex hull,
we first show that $\dH{\p}{u}$ contains all integer vectors on the boundary of its convex hull.
Suppose that $\bun$ is an integer vector on the boundary of $\Conv \dH{\p}{u}$.
\citet[Proposition 2.16]{BaKl:19} implies that there exists a price vector $\ppr$ such that
$\Conv \dH{\ppr}{u}$ has dimension $n$ and includes $\bun$.
The induction hypothesis implies that
$\bun \in \dH{\ppr}{u}.$
Thus, in the single-agent transferable utility economy with valuation $\valH{\cdot}{u}$ and endowment ${\bf w}^j = \bun^1$, the price vector $\ppr$ is the price vector of a competitive equilibrium and $\zero$ lies in the convexification of aggregate excess demand at $\p$;
Fact~\ref{fac:pseudoequilPrices} %
hence implies that
$\bun \in \dH{\p}{u}$.%
\footnote{\citet[Lemma 2.6]{BaKl:14} provides an alternative statement of Fact~\ref{fac:pseudoequilPrices} that allows one to bypass the construction of a single-agent economy.}

It remains to prove that $\dH{\p}{u}$ contains all integer vectors in the relative interior of its convex hull.
Let $\bun$ be an integer vector in the relative interior of $\Conv \dH{\p}{u}.$
We need to show that $\bun$ lies in $\dH{\p}{u}$.

As $\dH{\p}{u} \subseteq [0,M]^I$, there exists a set $I' \subseteq I$ of goods %
such that the rectangular prism %
\[P = \Conv \left\{\bun -  M \sum_{i \in I'} \e{i} + 2M \sum_{i \in I''} \e{i} \,\Bigg|\, I'' \subseteq I'\right\}\]
intersects $\Conv \dH{\p}{u}$ in a line segment,
say with vertices $\bun^1$ and $\bun^2$.
(Units of the goods in $I \ssm I'$ will be regarded as separate items.)
Since $\mathcal{D}$ is totally unimodular and the edges of $P$ are elementary basis vectors,
the set $P \cap \mathbb{Z}^I$ is $\mathcal{D}$-convex.
Since $\valH{\cdot}{u}$ is of demand type $\mathcal{D}$,
\citet[Proposition 2.16]{BaKl:19} implies that $\big(\Conv \dH{\p}{u}\big) \cap \mathbb{Z}^I$ is also $\D$-convex (see also \citet*[Footnote 35]{baldwin2020equilibrium}).
\citet*[Theorem 1]{danilov2004discrete} then guarantees that the polytope $\left(\Conv \dH{\p}{u}\right) \cap P$ has integer vertices.
Thus, $\bun^1$ and $\bun^2$ are integer vectors.
As $\bun^1$ and $\bun^2$ lie %
on the boundary of $\Conv \dH{\p}{u}$,
they must therefore lie in $\dH{\p}{u}$. %

\newcommand\doverline[1]{\overline{\overline{#1}}}
\newcommand\doutilFn{\doverline{U}^j}
\newcommand\doutil[2]{\doutilFn(#1)}
\newcommand\dodHFn{\dodQLFn_{\mathrm{H}}}
\newcommand\dodH[3]{\dodHFn(#1;#2)}
\newcommand\dodQLFn{\doverline{D}^j}
\newcommand\dodQL[2]{\odQLFn(#1)}

\newcommand\doutilFnMult\doutilFn

To complete the proof of the inductive step,
we first regard the units of items in $I \ssm I'$ as separate items,
then perturb prices of items from the price vector ${\bf p}$ to restrict demand to (a transformation of) the line segment between $\bun^1$ and $\bun^2$,
and conclude by using Lemma~\ref{lem:unitConsistentConsecutiveInteger} to show that demand must include $\bun$.

Formally, %
the set of items is $\doverline{I} = I' \cup ((I \ssm I') \times \{1,2,\ldots,M\})$.
Analogously to how we set up regarding all units of all goods as separate items in Section~\ref{sec:mult},
given a bundle $\doverline{\bun} \in \{0,1,\ldots,M\}^{I'} \times \{0,1\}^{(I \ssm I') \times \{1,2,\ldots,M\}}$ of items,
there is a corresponding bundle $\bun = \pi^{I \ssm I'}(\doverline{\bun})$ of goods
defined by
\[\pi^{I \ssm I'}(\doverline{\bun})_i = \begin{cases}
\sum_{m=1}^M \doverline{\bun}_{(i,m)} & \text{if } i \in I \ssm I'\\
\doverline{\bun}_i & \text{if } i \in I'
\end{cases}.\]
Consider the utility function $\doutilFnMult: \mathbb{R} \times (\pi^{I \ssm I'})^{-1}(X^j) \to \mathbb{R}$
for items and money defined by
\[\doutil{x_0,\doverline{\bun}}{I \ssm I'} = \util{x_0,\pi^{I \ssm I'}(\doverline{\bun})}.\]
We also consider a right inverse $\tau^{I \ssm I'}$ to $\pi^{I \ssm I'},$
which given a bundle $\buny \in \{1,2,\ldots,M\}^I$ of goods,
recovers a corresponding bundle of items by taking the lowest indices of items:
\begin{align*}
\tau^{I \ssm I'}(\buny)_{(i,m)} &= 1 & \text{for } i \in I \ssm I' \text{ and } m \le y_i\phantom{.}\\
\tau^{I \ssm I'}(\buny)_{(i,m)} &= 0 & \text{for } i \in I \ssm I' \text{ and } m > y_i\phantom{.}\\
\tau^{I \ssm I'}(\buny)_i &= y_i & \text{for } i \in I'.
\end{align*}

Due to the finiteness of $X^j$ and since $\bun^1 \in \dH{\p}{u}$, there exists $\varepsilon > 0$ such that
\[\valH{\buny}{u} - \p \cdot \buny < \valH{\bun^1}{u} - \p \cdot \bun^1 - M|I| \varepsilon
\quad \text{for all }\buny \notin \dH{\p}{u}.\]
Consider the price vector $\doverline{\p} \in \mathbb{R}^{\doverline{I}}$ defined by
\begin{align*}
\doverline{p}_i &= \pComp{i} & \text{for } i \in I'\phantom{.}\\
\doverline{p}_{(i,m)} &= \pComp{i} - \varepsilon & \text{for } i \in I \ssm I' \text{ and } m \le \bunComp{i}\phantom{.}\\
\doverline{p}_{(i,m)} &= \pComp{i} + \varepsilon & \text{for } i \in I \ssm I' \text{ and } m > \bunComp{i};
\end{align*}
intuitively,
this price vector perturbs ${\bf p}$ to make $\tau^{I \ssm I'}(\bun)$ is the unique cheapest bundle of items that recovers $\bun$.
By construction, we have that $\pi^{I \ssm I'}(\dodH{\doverline{\p}}{u}{I \ssm I'}) \subseteq \dH{\p}{u}$.
Moreover, due to the definition of expenditure minimization, we have that
\[\dodH{\doverline{\p}}{u}{I \ssm I'} = \argmin_{\doverline{\bun} \in (\pi^{I \ssm I'})^{-1}(\dH{\p}{u})} \left\{\sum_{i \in I'} \sum_{m=1}^M \doverline{x}_{(i,m)} (\doverline{p}_{i,m} - \pComp{i})\right\}.\]
The optimum of the minimization problem is $-\varepsilon \sum_{i \in I'} \bunComp{i}$,
which is achieved for example by $\tau^{I \ssm I'}(\bun^1)$ and $\tau^{I \ssm I'}(\bun^2)$.
The construction of $\doverline{\p}$ as a perturbation from ${\bf p}$ that lowers the prices of the first $\bunComp{i}$ units of each good $i \in I \ssm I'$ and raises the prices of the other units of each good $i \in I \ssm I$' hence ensures that
\begin{equation}
\label{eq:doverlinePartialLineSegment}
\dodH{\doverline{\p}}{u}{I \ssm I'} = \left\{\doverline{\bun} \in (\pi^{I \ssm I'})^{-1}(\dH{\p}{u}) \,\Big|\, \doverline{x}_{(i,m)} = \mathbbm{1}_{m \le \bunComp{i}} \text{ for } i \in I' \text{ and } 1 \le m \le M\right\}.
\end{equation}

It follows from the definitions of $\bun^1,\bun^2$ that
\[\dodH{\doverline{\p}}{u}{I \ssm I'} = (\pi^{I \ssm I'})^{-1}(\dH{\p}{u}) \cap \Conv \{\tau^{I \ssm I'}(\bun^1),\tau^{I \ssm I'}(\bun^2)\}.\]
Note that $\tau^{I \ssm I'}(\bun) \in \Conv \{\tau^{I \ssm I'}(\bun^1),\tau^{I \ssm I'}(\bun^2)\} = \Conv \dodH{\doverline{\p}}{u}{I \ssm I'}$ by construction,
but that $\tau^{I \ssm I'}(\bun) \notin \dodH{\doverline{\p}}{u}{I \ssm I'}$ since $\bun \notin \dH{\p}{u}$.

Let $i$ be a good with $\bunComp{i}^1 \not= \bunComp{i}^2$.
By construction, we must have that $i \notin I'$,
and $\bunComp{i}$ must lie between $\bunComp{i}^1$ and $\bunComp{i}^2.$
Note also since $\utilFn$ is unit-consistent,
$\doutilFnMult$ is unit-consistent as well.
Hence, Lemma~\ref{lem:unitConsistentConsecutiveInteger}
implies that there exists $\doverline{\bun} \in \dodH{\doverline{\p}}{u}{I \ssm I'}$ with $\doverline{x}_i = \bunComp{i}$.

\eqref{eq:doverlinePartialLineSegment} then implies that $\pi^{I \ssm I'}(\doverline{\bun}) \in \dH{\p}{u} \cap P$.
In particular, $\pi^{I \ssm I'}(\doverline{\bun})$ lies on the line segment between $\bun^1$ and $\bun^2$.
But since $\bunComp{i}^1 \not= \bunComp{i}^2$ and $\pi^{I \ssm I'}(\doverline{\bun})_i = \doverline{x}_i = \bunComp{i},$
we must therefore have that $\bun = \pi^{I \ssm I'}(\doverline{\bun}) \in \dH{\p}{u}.$
This completes the proof of the inductive step.

\subsection{Proof of the ``if'' directions of Propositions~\ref{prop:dkmRelationshipConsistentToDc} and~\texorpdfstring{\ref{prop:demTypeRelationshipConsistentToUnimod}}{\ref{prop:dkmRelationshipConsistentToDc}'}}

Let $S$ be the set of all compensated price effects for all agents and goods that lie in $\{-1,0,1\}^I$,
and let $\D = S \cup -S$.
The ``only if'' direction of Proposition~\ref{prop:totallyunimodgen} implies that $S$ (and hence $\D$) is totally unimodular.
By the ``if'' direction of Proposition~\ref{prop:DquasiconcaveEquiv},%
\footnote{\label{fn:dkmRelationship}To prove the ``if'' directions of Propositions~\ref{prop:dkmRelationshipConsistentToDc} and~\ref{prop:demTypeRelationshipConsistentToUnimod} in the case in which each agent demand at most one unit of each good,
one can use \citet*[Theorem 1]{baldwin2021consumer} instead of the ``if'' direction of Proposition~\ref{prop:DquasiconcaveEquiv} here.}
all agents' preferences are $\D$-quasiconcave---hence, belong to a single class of discrete~convexity.

\subsection{Proof of Theorems~\ref{thm:existBinary} and~\ref{thm:existGen}}

Theorem~\ref{thm:existBinary} is a special case of Theorem~\ref{thm:existGen},
so it suffices to prove Theorem~\ref{thm:existGen}.

By the ``if'' direction of Proposition~\ref{prop:dkmRelationshipConsistentToDc},
the preferences must lie in a single class of discrete convexity.
\citet*[Theorem 2 and 4]{DaKoMu:01} or \citet*[Theorem 3]{baldwin2020equilibrium} therefore guarantees that competitive equilibria exists.

\subsection{Proof of Theorems~\ref{thm:neccBinary} and~\ref{thm:neccGen}}

Theorem~\ref{thm:neccBinary} is a special case of Theorem~\ref{thm:neccGen},
so it suffices to prove Theorem~\ref{thm:neccGen}.

Consider an invariant domain $\mathcal{V}$.
We prove Theorem~\ref{thm:neccGen} in three steps.%
\footnote{\label{fn:binaryNecc}To prove Theorem~\ref{thm:neccBinary} alone, one can skip the first step;
by using \citet*[Theorem 1]{baldwin2021consumer} instead of Proposition~\ref{prop:demTypeEquiv}\ref{part:demTypeOneDirectionConsistentEquiv}, one can also circumvent the second step.}
\begin{enumerate}
\item We first show that if competitive equilibria exist in all economies in which agents have valuations $\mathcal{V}$,
then the valuations in $\mathcal{V}$ must be unit-consistent.
\item We then show that if the valuations in $\mathcal{V}$ are unit-consistent and
competitive equilibria exist in all economies in which agents have valuations in $\mathcal{V}$,
then each pair of goods must be consistent---i.e., either complements or substitutes.
\item We next show that if the valuations in $\mathcal{V}$ are unit-consistent, each pair of goods is consistent, and
competitive equilibria exist in all economies in which agents have valuations in $\mathcal{V}$,
then the valuations in $\mathcal{V}$ must be bundle-consistent.
\end{enumerate}

\newcommand\ovalFn{\overline{V}^j}

\paragraph*{Necessity of unit consistency.}
The proof of this step is similar in spirit to the proof of \citet[Theorem 16]{MiSt:09}, but somewhat more involved since we allow for complementarities across goods.

Suppose for sake of deriving a contradiction that there exists a valuation $\valFn \in \mathcal{V}$ that is not unit-consistent.
Applying the \ref{state:demType}$\implies$\ref{state:ourDef} implication of Lemma~\ref{lem:equivDefs} %
to the corresponding valuation $\ovalFn$ when each unit of each good is regarded as a separate item,
we see that there exists a price vector $\overline{\p} \in \mathbb{R}^{I \times \{1,2,\ldots,M\}}$
such that $\Conv \odQL{\p}$ is one-dimensional and parallel to a vector $\dvec$ with positive components for two items that correspond to units of the same good.
By invariance,
by adding a linear valuation to $\valFn$ if necessary,
we can assume that $\overline{\p} \in \mathbb{R}_{> 0}^{I \times M}$.
Let $\varepsilon = \frac{1}{2} \min_{i,m} \overline{p}_{i,m}.$

\citet*[Proposition 2.20]{BaKl:19} implies that by perturbing $\p$ and permuting the labels of items that correspond to the units of the same good if necessary,
we can assume that $\overline{p}_{i,m} < \overline{p}_{i,m'}$ for all goods $i$ and all $m < m'$.
Defining a right inverse to $\pi,$
which given a bundle $\bun \in \{1,2,\ldots,M\}^I$ of goods,
associates a bundle of items defined by
\[\tau(\bun)_{(i,m)} = \begin{cases}
1 & \text{if } m \le \bunComp{i}\\
0 & \text{if } m \ge \bunComp{i}
\end{cases},\]
we see that $\overline{\p} \cdot \overline{\bun} \ge \overline{\p} \cdot \tau(\pi(\overline{\bun}))$ with equality if and only if $\overline{\bun} = \tau(\pi(\overline{\bun})).$
Hence,
due to the construction of $\ovalFn$,
we have that $\overline{\bun} = \tau(\pi(\overline{\bun}))$ for all $\overline{\bun} \in \odQL{\p}$.

Since $\Conv \odQL{\overline{\p}}$ is one-dimensional and $\odQL{\overline{\p}} \subseteq \{0,1\}^{I \times \{1,2,\ldots,M\}},$
we must have that $|\odQL{\p}| = 2$.
Let $\odQL{\overline{\p}} = \{\tau(\bun),\tau(\bunpr)\}$.
Due to the hypothesis that $\dvec$ has positive components for two items that correspond to units of the same good,
by exchanging $\bun,\bunpr$ if necessary,
we can ensure that there exists a good $k$ with $\bunprComp{k} - \bunComp{k} \ge 2$.
Consider the endowment of goods ${\bf w}^j = \bun$
and the endowment of items $\overline{{\bf w}}^j = \tau(\bun)$.

For each $1 \le m \le M,$ consider the price vector $\p^m \in \mathbb{R}^I$ defined by
\[p^m_i = \begin{cases}
\overline{p}_{(i,m)} & \text{if } i \not= k \text{ or } m > \bunprComp{k}\\
\varepsilon & \text{if } i = k \text{ and } m < \bunprComp{k}\\
0 & \text{if } i = k \text{ and } m = \bunprComp{k}
\end{cases}\]
and consider the valuation for goods $V^m: \{0,1\}^I \to \mathbb{R}$ defined by $V^m(\bun) = \p^m \cdot \bun$.
Last, consider the valuation for goods $V^{M+1}: \{0,1\}^I \to \mathbb{R}$ defined by $V^{M+1}(\bun) = \overline{p}_{(k,\bunprComp{k})} \bunComp{k}$.
Consider the endowment of goods ${\bf w}^{M+1} = \zero$,
and, for $1 \le m \le M$, the endowment of goods ${\bf w}^m$ defined by
\[w^m_i = \begin{cases}
1 & \text{if } m > \bunComp{i}\\
0 & \text{otherwise}
\end{cases}.\]

For $1 \le m \le M,$
consider the valuation for items $\tilde{V}^m$ with domain \[\tilde{X}^m = \{0,1\}^{I \times \{m\}} \times \{0\}^{I \times \{1,2,\ldots,m-1,m+1,\ldots,M\}}\]
defined by $\tilde{V}^m(\overline{\bun}) = V^m(\pi(\overline{\bun}))$.
Consider also the valuation for items $\tilde{V}^{M+1}$ with domain $\tilde{X}^{M+1} = \{\zero\} \cup \{\e{(k,m)} \mid \bunComp{k} < m \le \bunprComp{k}\}$
defined by $\tilde{V}^{M+1}(\overline{\bun}) = \overline{\p} \cdot \overline{\bun}$.
Consider the endowments of items $\overline{{\bf w}}^{M+1} = \zero$,
and, for $1 \le m \le M$, the endowment of items $\overline{{\bf w}}^m$ defined by
\[\overline{w}^m_{(i,m')} = \begin{cases}
w^m_{(i,m')} & \text{if } m' = m\\
0 & \text{otherwise}
\end{cases}.\]

\newcommand\hobun{\hat{\overline{\bun}}}
\newcommand\hobunComp[1]{\hat{\overline{x}}_{#1}}
\newcommand\hop{\hat{\overline{\p}}}
\newcommand\hopComp[1]{\hat{\overline{p}}_{#1}}

Consider the auxiliary economy with set of goods $I \times M$ and agents $\{1,\ldots,M+1\} \cup \{j\}$
and valuations $\tilde{V}^1,\ldots,\tilde{V}^{M+1},\tilde{V}^j$,
where we write $\tilde{V}^j = \overline{V}^j$.
At price vector $\overline{\p}$,
the aggregate excess demand of agents $1,2,\ldots,M$ is
\[\left\{\overline{\bun} \in \{-1,0,1\}^{I \times M} \,\left|\, \begin{array}{l}
\overline{x}_{(k,m)} = 0 \text{ for } m \le \bunComp{k}, \overline{x}_{(k,m)} = -1 \text{ for } \bunComp{k} < m \le \bunprComp{k} ,\\
\overline{x}_{(i,m)} \ge 0 \text{ for } m \le \bunComp{i}, \text{ and } \overline{x}_{(i,m)} \le 0 \text{ for } m > \bunComp{i}
\end{array}\right.
\right\}.\]
At the same price vector,
excess demand of agent $M+1$ is $\{\e{(k,m)} \mid \bunComp{k} < m \le \bunprComp{k}\}$,
and excess demand of agent $j$ is $\{\zero,\tau(\bunpr) - \tau(\bun)\}$.
Note also that
\begingroup
\allowdisplaybreaks
\begin{align*}
\zero =& \left[\frac{1}{\bunprComp{k} - \bunComp{k}} \left(- \sum_{m = \bunComp{k}}^{\bunprComp{k}} \e{(k,m)}\right) + \left(1-\frac{1}{\bunprComp{k} - \bunComp{k}}\right) \left(\tau(\bun) - \tau(\bun')\right)\right]\\
& + \left[\sum_{m=\bunComp{k}}^{\bunprComp{k}} \frac{1}{\bunprComp{k} - \bunComp{k}} \e{(k,m)}\right] + \left[\frac{1}{\bunprComp{k} - \bunComp{k}}\zero + \left(1-\frac{1}{\bunprComp{k} - \bunComp{k}}\right) (\tau(\bunpr) - \tau(\bun))\vphantom{\sum_{m=\bunComp{k}}^{\bunprComp{k}}}\right].
\end{align*}
\endgroup
The vectors $- \sum_{m = \bunComp{k}}^{\bunprComp{k}} \e{(k,m)}$ and $\tau(\bun) - \tau(\bun')$ both lie in the aggregate excess demand of agents $1,2,\ldots,M$ by construction.
Therefore, the convexification of aggregate excess demand at $\overline{\p}$ contains $\zero$.
However, at each selection from aggregate excess demand at $\overline{\p}$, the total aggregate excess demand for items $(k,\bunComp{k}+1),(k,\bunComp{k}+2),\ldots,(k,\bunprComp{k})$ is either $1$ or $1 + \bunComp{k} - \bunprComp{k} \le -1$.
In particular, $\overline{\p}$ is an not a competitive equilibrium price vector.
Hence, by the contrapositive of Fact~\ref{fac:pseudoequilPrices},
no competitive equilibria exist in the auxiliary economy. 

To obtain a contradiction, we claim that competitive equilibrium in fact exist in the auxiliary economy.
To show this, we consider the original economy with set of goods $I$ and agents $\{0,1\ldots,M+1\} \cup \{j\}$.
The invariance of $\mathcal{V}$ ensures that $V^m \in \mathcal{V}$ for $0 \le m \le M+1$.
Hence, by hypothesis,
there exists a competitive equilibrium, say with allocation
$\hbun^0,\hbun^1,\ldots,\hbun^{M+1},\hbun^j$ and price vector $\hp$.
As
\[w^j_i + \sum_{m=1}^{M+1} w^m_i = M\]
for all goods $i$ by construction,
and $V^1,\ldots,V^M$ each place value at least $\varepsilon$ on a first unit of each good $i \not= k$,
the First Welfare Theorem implies that $\hbunComp{i}^{M+1} = 0$ for $i \not= k$.
Moreover,
as $V^1,\ldots,V^{\bunprComp{k}-1},V^{\bunprComp{k}+1},\ldots,V^{M+1}$ each place value at least $\varepsilon$ on a first unit of good $k$,
the First Welfare Theorem also implies that $\hbunComp{k}^{\bunprComp{k}} = 0$.

Consider the price vector $\hop \in \mathbb{R}^{I \times M}$ defined by $\hopComp{(i,m)} = \hpComp{i}$.
Consider also the allocation in the auxiliary economy defined by
$\hobun^{M+1} = \hobunComp{k}^{M+1} \e{(k,\bunprComp{k})}$; for $1 \le m \le M$, the consumption vector of items $\hobun^m$ defined by
\[\hobunComp{(i,m')}^m = \begin{cases}
\hbunComp{i}^m & \text{if } m' = m\\
0 & \text{otherwise}
\end{cases};\]
and the consumption vector $\hobun^j$ of items defined by
\[\hobunComp{(i,m)}^j = \begin{cases}
1 - \hobunComp{(i,m)}^{M+1} & \text{if } i = k \text{ and } m = \bunprComp{k}\\
1 - \hobunComp{(i,m)}^m & \text{otherwise}
\end{cases}.\]

We claim that the allocation $\hbun^1,\ldots,\hbun^{M+1},\hbun^j$ and the price vector $\hop$ comprise a competitive equilibrium in the auxiliary economy.
Since $\hbunComp{i}^{M+1} = 0$ for $i \not= k$,
we have that $\pi(\hobun^m) = \hbun^m$ for all $1 \le m \le M+1$.
We also have that $\hbun^m \in \tilde{X}^m$ for all $1 \le m \le M+1$.
In particular, as $\tilde{V}^m(\overline{\bun}) = V^m(\pi(\overline{\bun}))$ for all $1 \le m \le M+1$ and $\overline{\bun} \in \tilde{X}^m$,
it follows that $\hobun^m \in \tilde{D}^m(\hop)$ for $1 \le m \le M+1$.
Since furthermore $\hbunComp{k}^{\bunprComp{k}} = 0$, we have that
\[\hobunComp{(i,m')}^j + \sum_{m=1}^{M+1} \hobunComp{(i,m')}^m = 1 = \overline{w}^j_{(i,m')} + \sum_{m=1}^{M+1} \overline{w}^m_{(i,m')}\]
for all items $(i,m')$.
As market-clearing in the original economy entails that
\[\hbunComp{i}^j + \sum_{m=1}^{M+1} \hbunComp{i}^m = w^j_i + \sum_{m=1}^{M+1} w^m_i\]
for all goods $i$,
it follows that $\pi(\hobun^j) = \hbun^j$,
and hence that $\hobun^j \in \tilde{D}^j(\hop)$.
Therefore, the allocation $\hbun^1,\ldots,\hbun^{M+1},\hbun^j$ and the price vector $\hop$ comprise a competitive equilibrium in the auxiliary economy---contradicting the non-existence of competitive equilibria in the auxiliary economy (obtained at an earlier stage of the argument).
Hence, we can conclude that the valuations in $\mathcal{V}$ must in fact be unit-consistent.

\paragraph*{Necessity of consistency.}
This proof of this step is based on \citeposs{KeCr:82} example of how inconsistency between goods can obstruct equilibrium existence,
as generalized by \cite{GuSt:99} to prove a ``maximal domain'' result for substitutes.

We prove the contrapositive.
Suppose that $\mathcal{V}$ is unit-consistent, but goods $k$ and $\ell$ are inconsistent;
we show that there exist valuations in $\mathcal{V}$ and endowments for which competitive equilibria do not exist.

Let $V^1 \in \mathcal{V}$ (resp.~$V^2 \in \mathcal{V}$) be a valuation for which goods $k$ and $\ell$ are not substitutable (resp.~not complementary) given bundling $\B$.
By the \ref{state:demType}$\implies$\ref{state:ourDef} implication of Lemma~\ref{lem:equivDefs}, if $V^1$ (resp.~$ V^2$) is of demand type $\D$,
there must exist $\dvec \in \D$ with $\dvecComp{k}\dvecComp{\ell} > 0$ (resp.~$\dvecComp{k}\dvecComp{\ell} < 0$).
Thus, there exists a price vector $\p^1$ (resp. $\p^2$) such that $\Conv D^1(\p^1)$ (resp.~$\Conv D^2(\p^2)$) is one-dimensional and parallel to an integer vector $\dvec^1$ (resp.~$\dvec^2$) with $\dvecComp{k}^1\dvecComp{\ell}^1 > 0$ (resp.~$\dvecComp{k}^2\dvecComp{\ell}^2 < 0$).
By invariance,
we can add linear functions to $V^1,V^2$ to ensure that $\p^1 = \p^2 = \p \in \mathbb{R}_{>0}^I$.

Without loss of generality,
we can assume that the components of $\dvec^1$ (resp. $\dvec^2$) have no common factors, and that $\dvecComp{k}^1, \dvecComp{k}^2> 0$.
Proposition~\ref{prop:demTypeEquiv}\ref{part:demTypeOneDirectionNotConsistent} implies that $\dvec^1,\dvec^2 \in \{-1,0,1\}^I$.
In particular, we must have that $\dvecComp{k}^1 = \dvecComp{\ell}^1 = \dvecComp{k}^2 = 1$ and that $\dvecComp{\ell}^2 = -1$.

For $1 \le j \le 2$,
let one endpoint of $\Conv D^j(\p)$ be $\bundowj$ such that $\Conv \dQL{\p}$ is a subset of the ray with endpoint $\bundowj$ and slope $\dvec^j$.
Consider also the valuation $V^{3}: \{0,1\}^I \to \mathbb{R}$ defined by $V^3(\bun) = \hp \cdot \bun,$
where $\hp = (0_{k,\ell},\p_{I \ssm \{k,\ell\}}),$
and the consumption vector
\[\bundow^{3} = \begin{cases}
1 & \text{if } \dvecComp{i}^1 + \dvecComp{i}^2 > 0\\
0 & \text{if } \dvecComp{i}^1 + \dvecComp{i}^2 \le 0
\end{cases}.\]
The valuation $V^{3}$ is in the domain $\mathcal{V}$ due to the invariance of $\mathcal{V}$.

Consider the economy with agents $\{1,2,3\}$.
By construction,
excess demand for agent 1 (resp.~agent 2) at the price vector $\p$ is a subset of $\mathbb{R}_{\ge 0} \dvec^1$ (resp.~$\mathbb{R}_{\ge 0} \dvec^2$) that includes at least one nonzero integer vector.
In particular, $\frac{1}{2}\dvec^1$ (resp.~$\frac{1}{2}\dvec^2$) lies in the convex hull of excess demand for agent 1 (resp.~agent 2).
Since $\dvecComp{k}^1 = \dvecComp{\ell}^1 = \dvecComp{k}^2 = 1$ and that $\dvecComp{\ell}^2 = -1$,
agent 3's excess demand at $\p$ is
\[\left\{\bun \in \{-1,0,1\}^I \mid \bunComp{k} = -1, \, \bunComp{\ell} = 0, \, \bunComp{i} \le 0 \text{ for } \dvecComp{i}^1 + \dvecComp{i}^2 > 0, \, \text{and } \bunComp{i} \ge 0 \text{ for } \dvecComp{i}^1 + \dvecComp{i}^2 \le 0\right\}.\]
In particular, $\frac{-\dvec^1-\dvec^2}{2}$ lies in the convex hull of excess demand for agent 3.
Thus, $\zero$ lies in the convex hull of aggregate excess demand at price vector $\p$.

However, since $\dvecComp{k}^1 = \dvecComp{\ell}^1 = \dvecComp{k}^2 = 1$,
at each selection from excess demand at $\p$ for agent 1 (resp.~agent 2, agent 3), the total excess demand for goods $k$ and $\ell$ is even (resp.~even, odd).
In particular, $\p$ is an not a competitive equilibrium price vector.
Hence,
by the contrapositive of Fact~\ref{fac:pseudoequilPrices},
no competitive equilibria can exist.

\paragraph*{Necessity of bundle consistency.}
The proof of this step is based on applying a similar construction to the previous step,
but where preferences for agents other than the two exhibiting a bundle inconsistency are constructed to be effectively linear in bundles in the bundling that exhibits a bundle inconsistency, rather than additive.

We prove the contrapositive.
Suppose that $\mathcal{V}$ is unit-consistent and satisfies the property that each pair of goods is consistent, but that $\mathcal{V}$ is not bundle-consistent;
we show that there exist valuations in $\mathcal{V}$ and endowments for which competitive equilibria do not exist.

Let $\B$ be a bundling consisting of relevant bundles that contains two bundles that are inconsistent.
Among all such bundlings,
consider one that includes as many elementary basis vectors as possible.
Let $\B = \{\b^1,\ldots,\b^{|I|}\}$ where $\b^1$ and $\b^2$ are inconsistent bundles.

Consider the matrix $G$ whose columns are $\b^1,\ldots,\b^{|I|}$ in that order.
It follows from Lemma~\ref{lem:nonUnimodToInconsistent} that $\B$ is totally unimodular.
Hence, the determinant of $G$ is $\pm 1$.

Let $V^1 \in \mathcal{V}$ (resp.~$V^2 \in \mathcal{V}$) be a valuation for which $\b^1$ and $\b^2$ are not substitutable (resp.~complementary) given bundling $\B$.
Given a valuation $\valFn$,
define a transformed $G^*\valFn: G^{-1} X^j \to \mathbb{R}$ by $(G^*\valFn)(\q) = \val{G \q}$
\citep[page 885]{BaKl:19}.
By construction, demand for $G^*\valFn$ at price vector $\tp$ is $\dB{\p}{\B}$. 
Hence,
the first and second goods are not substitutes (resp.~complements) for $G^*V^1$ (resp.~$G^*V^2$).
By the \ref{state:demType}$\implies$\ref{state:ourDef} implication of Lemma~\ref{lem:equivDefs}, if $G^*V^1$ (resp.~$G^* V^2$) is of demand type $\D$,
there must exist $\dvec \in \D$ with $\dvecComp{1}\dvecComp{2} > 0$ (resp.~$\dvecComp{1}\dvecComp{2} < 0$).
\citet[Proposition 3.11]{BaKl:19} implies that $\valFn$ is of demand type $\D$ if and only if that $G^*\valFn$ is of demand type $\D$.
Hence, if $V^1$ (resp.~$V^2$) is of demand type $\D$,
then there must exist $\dvec \in \D$ with $(G^{-1} \dvec)_1 (G^{-1} \dvec)_2 > 0$ (resp.~$(G^{-1} \dvec)_1 (G^{-1} \dvec)_2 < 0$).
Thus, there exists a price vector $\p^1$ (resp. $\p^2$) such that $\Conv D^1(\p^1)$ (resp.~$\Conv D^2(\p^2)$) is one-dimensional and parallel to an integer vector $\dvec^1$ (resp.~$\dvec^2$) whose components have no non-trivial common factors such that $(G^{-1} \dvec^1)_1 (G^{-1} \dvec^1)_2 > 0$ (resp.~$(G^{-1} \dvec^2)_1 (G^{-1} \dvec^2)_2 < 0$).

Without loss of generality,
we can assume %
that $(G^{-1} \dvec^1)_1, (G^{-1} \dvec^2)_1> 0$.
By Proposition~\ref{prop:demTypeEquiv}\ref{part:demTypeOneDirectionNotConsistent}, $\pm \dvec^1$ is a price effect for $V^1 \in \mathcal{V}$ and an element of $\{-1,0,1\}^I$.
Due to the choice of $\B$,
applying the contrapositive of Lemma~\ref{lem:nonUnimodToInconsistent} to $S = (\B \cup \{\dvec^1\}) \ssm \{\e{i} \mid i \in I\}$,
we see that $\B \cup \{\dvec^1\}$ must be unimodular.
Applying $G^{-1}$,
we see that $\{\e{1},\ldots,\e{|I|},G^{-1} \dvec^1\}$ is  unimodular.
Hence, the components of $G^{-1} \dvec^1$ must each be $0$ or $\pm 1$.
Similarly, the components of $G^{-1} \dvec^2$ must each be $0$ or $\pm 1$.

In particular, we must have that $(G^{-1} \dvec^1)_1 = (G^{-1} \dvec^1)_2 = (G^{-1} \dvec^2)_1 = 1$ and that $(G^{-1} \dvec^2)_2 = -1$.
It follows that
\[\frac{1}{2} G^{-1} \dvec^1 + \frac{1}{2} G^{-1} \dvec^2 + \sum_{j = 3}^{|I|} \frac{-\left[(G^{-1}\dvec^1)_j + (G^{-1}\dvec^2)_j\right]}{2} \e{j} = \e{1}.\]
Multiplying both sides by $G$ on the left yields that
\begin{equation}
\label{eq:neccGenBundConsistentPseudoEquil}
\frac{1}{2} \dvec^1 + \frac{1}{2} \dvec^2 + \sum_{j = 3}^{|I|} \frac{-\left[(G^{-1}\dvec^1)_j + (G^{-1}\dvec^2)_j\right]}{2} \b^j = \b^1,
\end{equation}
where we have $\left|\frac{(G^{-1}\dvec^1)_j + (G^{-1}\dvec^2)_j}{2}\right| \le 1$ for $3 \le j \le |I|$ since $|(G^{-1}\dvec^1)_j|,|(G^{-1}\dvec^2)_j| \le 1.$

Since $\mathcal{V}$ contains the zero valuation $V_0$,
each bundle in $\B$ corresponds to a price effect for some valuation in $\mathcal{V}$.
For $3 \le j \le |I|$,
consider a valuation $\valFn$ for which $\b^j$ is a price effect.
By Proposition~\ref{prop:demTypeEquiv}\ref{part:demTypeOneDirectionConsistentEquiv},
$\valFn$ is of demand type $\D$ only if $\b^j \in \D$.
Thus, for $3 \le j \le |I|$,
there exists a price vector $\p^j$ such that $\Conv \dQL{\p^j}$ is one-dimensional and parallel to $\b^j$.
By invariance,
we can add linear functions to $V^j$ to ensure that $\p^j = \p \in \mathbb{R}_{>0}^I$ for $1 \le j \le |I|$.

For $1 \le j \le 2$,
let one endpoint of $\Conv D^j(\p)$ be $\bundowj$ such that $\Conv \dQL{\p}$ is a subset of the ray with endpoint $\bundowj$ and slope $\dvec^j$.
For $3 \le j \le |I|$, let one endpoint of $\Conv \dQL{\p}$ be $\bundowj$ such that $\Conv \dQL{\p}$ is a subset of the ray with endpoint $\bundowj$ and slope $-\b^j$ (resp.~$\b^j$)
if $(G^{-1}\dvec^1)_j + (G^{-1}\dvec^2)_j \ge 0$ (resp.~$(G^{-1}\dvec^1)_j + (G^{-1}\dvec^2)_j < 0$).
Consider also the valuation $V^{|I|+1}: \{0,1\}^I \to \mathbb{R}$ defined by $V^{|I|+1}(\bun) = \hp \cdot \bun,$
where
\[\hpComp{i} = \begin{cases}
2\pComp{i} & \text{if } \bComp{i}^1 \ge 0\\
0 & \text{if } \bComp{i}^1 < 0
\end{cases}\]
and the consumption vector
\[\bundow^{|I|+1} = \begin{cases}
1 & \text{if } \bComp{i}^1 > 0\\
0 & \text{if } \bComp{i}^1 \le 0
\end{cases}.\]
The valuation $V^{|I|+1}$ is in the domain $\mathcal{V}$ due to the invariance of $\mathcal{V}$.

Consider the economy with agents $j = 1,2,\ldots,|I|+1$,
where agent $j$ has valuation $\valFn$ and endowment $\bundowj$ of indivisible goods.
By construction, at the price vector $\p$:
\begin{itemize}
\item excess demand for agent 1 (resp.~agent 2) is a subset of $\mathbb{R}_{\ge 0} \dvec^1$ (resp.~$\mathbb{R}_{\ge 0} \dvec^2$) that includes at least one nonzero integer vector,
\item if $(G^{-1}\dvec^1)_j + (G^{-1}\dvec^2)_j \ge 0$, excess demand for agent $3 \le j \le |I|$ is a subset of $\mathbb{R}_{\le 0} \b^j$ that includes at least nonzero integer vector;
\item if $(G^{-1}\dvec^1)_j + (G^{-1}\dvec^2)_j < 0$, excess demand for agent $3 \le j \le |I|$ is a subset of $\mathbb{R}_{\ge 0} \b^j$ that includes at least nonzero integer vector; and
\item excess demand for agent $|I|+1$ is $-\b^1$.
\end{itemize}
Hence, it follows from \eqref{eq:neccGenBundConsistentPseudoEquil} that $\zero$ lies in the convex hull of aggregate excess demand at price vector $\p$.
However, as the vectors $\dvec^1,\dvec^2,\b^3,\b^4,\ldots,\b^{|I|}$ are linearly independent and the summands $\frac{1}{2} \dvec^1$ and $\frac{1}{2} \dvec^2$ are not integer,
$\zero$ does not lie in aggregate excess demand at price vector $\p$.
Hence,
by the contrapositive of Fact~\ref{fac:pseudoequilPrices},
no competitive equilibria can exist.

\subsection{Proof of the ``only if'' directions of Propositions~\ref{prop:dkmRelationshipConsistentToDc} and~\texorpdfstring{\ref{prop:demTypeRelationshipConsistentToUnimod}}{\ref{prop:dkmRelationshipConsistentToDc}'}}

Let $\D$ be totally unimodular.
Then, $\D' = \D \cup \{\pm \e{i} \mid i \in I\}$ is also totally unimodular.
As $\D'$ contains the elementary basis vectors,
the zero valuation $V_0$ is $\D'$-quasiconcave.
Hence, the class of $\D'$-quasiconcave utility functions is invariant.
Since $\D'$-quasiconcave utility functions form a domain for equilibrium existence (\citealp*[Theorems 2 and 4]{DaKoMu:01}; \citealp*[Theorem 3]{baldwin2020equilibrium}).
Hence, Theorem~\ref{thm:neccGen} implies that the class of all $\D'$-quasiconcave utility functions is unit- and bundle-consistent.
As $\D' \supseteq \D$, the class of all $\D$-quasiconcave utility functions is unit- and bundle-consistent.
As $\D$ was arbitrary, we have proved that every class of discrete convexity is unit- and bundle-consistent.

\subsection{Proof of the ``if'' direction of Proposition~\ref{prop:totallyunimodgen}}

Let $S$ be the set of all compensated price effects for all agents and goods that lie in $\{-1,0,1\}^I$,
and let $\D = S \cup -S$, which is totally unimodular by hypothesis.
The ``if'' direction of Proposition~\ref{prop:DquasiconcaveEquiv} implies that preferences are all $\D$-quasiconcave.
The ``only if'' direction of Proposition~\ref{prop:dkmRelationshipConsistentToDc} then implies that preferences are bundle-consistent.

\subsection{Proof of the ``only if'' directions of Propositions~\ref{prop:DquasiconcaveEquiv} and~\texorpdfstring{\ref{prop:demTypeEquivTotallyUnimod}}{\ref{prop:DquasiconcaveEquiv}'}}

Consider a $\D$-quasiconcave utility function $\utilFn$.
Since $\utilFn$ is of demand type $\D$,
Fact~\ref{fac:EEDlemma1} and \citet*[Propositions 1 and D.2]{baldwin2021consumer} imply that each pair of goods is either complements or substitutes.
The ``only if'' direction of Proposition~\ref{prop:dkmRelationshipConsistentToDc} implies that $\utilFn$ is unit-consistent.
Therefore, Proposition~\ref{prop:demTypeEquiv}\ref{part:demTypeOneDirectionConsistentEquiv} implies that all compensated price effects for agent $j$ that lie in $\{-1,0,1\}^I$ must lie in~$\D$.

\clearpage

\begin{center}
{\LARGE Online Appendix}
\end{center}

\section{A necessity result with income effects}
\label{oapp:neccIncomeEffs}

In this appendix,
we extend Theorem~\ref{thm:neccGen} to settings with income effects.
First, we extend the definition of invariance to domains of utility functions.

\begin{primedefn}{def:invariance}
A domain $\mathcal{U}$ of utility functions is \emph{invariant} if:
\begin{itemize}
\item for all $U \in \mathcal{U}$ and all price vectors $\p \in \mathbb{R}_{\ge 0}^I$, writing $U'(x_0,\bun) = U(x_0 + \p \cdot \bun,\bun)$, we have that $U' \in \mathcal{U}$; and
\item $\mathcal{U}$ contains the quasilinear utility function with valuation $V_0: \{0,1\}^I \to \mathbb{R}$ defined by $V_0(\bun) = 0$.
\end{itemize}
\end{primedefn}

With this notion of invariance, Theorem~\ref{thm:neccGen} extends to settings with income effects.

\begin{thm}%
\label{thm:neccGenInc}
If competitive equilibria exist in all economies in which agents have utility functions in an invariant domain $\mathcal{U}$,
then the utility functions in $\mathcal{U}$ are unit- and bundle-consistent.
\end{thm}
\begin{proof}
The class $\mathcal{V}$ of Hicksian valuations for utility functions in $\mathcal{U}$ is invariant.
By \citet*[Theorem 1]{baldwin2021equilibrium},
competitive equilibria exist in all economies in which agents have valuations in $\mathcal{V}$.
Theorem~\ref{thm:neccGen} therefore implies that the valuations in $\mathcal{V}$ are unit- and bundle-consistent.
The theorem follows by \eedlemma.
\end{proof}

\section{Necessity of bundling goods with sale opportunities}
\label{oapp:egs}

In our leading Example~\ref{eg:threecycle} (revisited in Section~\ref{sec:leading} in the main text),
we bundled only (strictly complementary) goods in order to reveal a bundle inconsistency. The following example shows that in order to reveal a bundle inconsistency,
one may need to bundle goods with (strictly complementary) opportunities to sell other goods.

\begin{eg}
\label{eg:hiddencomp}

Let the set of agents be $I = \{\apple,\banana,\coconut,\date,\eld\}$ denote the set of goods, and let the set of agents be $J=\{1,2,3,4,5\}$. Suppose that agents' preferences are quasilinear and are given by the following valuations: 
\begin{gather*}
V^1(\bun)=3\min\{\bunComp{\ga}+\bunComp{\gb},1\},\quad
V^2(\bun)=3\min\{\bunComp{\gb}+\bunComp{\gc},1\},\quad
V^3(\bun)=3\min\{\bunComp{\gc}+\bunComp{\gd},1\},\\
V^4(\bun)=3\min\{\bunComp{\gd}+\bunComp{\gel},1\},\quad
V^5(\bun)=3\min\{\bunComp{\ga},\bunComp{\gel}\}.
\end{gather*}
The first four agents view $\apple$ and $\banana$, $\banana$ and $\coconut$, $\coconut$ and $\date$, $\date$ and $\eld$ as (perfect) substitutes respectively while agent 5 views   $\apple$ and $\eld$ as (perfect) complements.  

Consider the bundling $\B$ that includes every good except the $\eld$ as well as a bundle of $\apple$ and $\eld$.%
\footnote{Formally,
$\B=\{(1,0,0,0,0),(0,1,0,0,0),(0,0,1,0,0),(0,0,0,1,0),(1,0,0,0,1)\}.$}
This bundling only bundles goods. We can verify that this bundling does not create any inconsistencies. For example, consider bundle prices $\tp=(1,1,1,1,4)$. At those prices the bundled demands of agents~4 and 5 are:
$$\widetilde{D}^4(\tp;\B)=(0,0,0,1,0)\quad\quad \widetilde{D}^5(\tp;\B)=(0,0,0,0,0).$$
Suppose that the price of $\apple$ and $\eld$ bundle decreases so the new bundle prices are $\tp'=(1,1,1,1,1)$. The bundled demand is now:
$$\widetilde{D}^4(\tp';\B)=(-1,0,0,0,1)\quad\quad \widetilde{D}^5(\tp';\B)=(0,0,0,0,1).$$
Agent~5 now simply wants to buy the bundle. But now the bundle is cheap enough that agent~4 (who wants either the $\date$ or $\eld$) wants to buy the bundle and then sell the $\apple$ to end up with the $\eld$.
There are no inconsistencies under this bundling.

Consider instead the bundling $\B'$ in which $\apple$ and $\eld$ can be traded separately, but $\banana$ is bundled with an opportunity to sell the $\apple$, $\coconut$ is bundled with an opportunity to sell the $\banana$, and $\date$ is bundled with the opportunity to sell the $\coconut$.%
\footnote{%
Formally,
$\B'=\{(1,0,0,0,0),(-1,1,0,0,0),(0,-1,1,0,0),(0,0,-1,1,0),(0,0,0,0,1)\}.$}
We claim that there is then an inconsistency.  To see this, 
let us start at bundle prices $\tp=(0.5,0.5,0.5,0.5,4)$. At these prices, the bundled demands of agents~4 and~5 are
$$\widetilde{D}^4(\tp;\B)=(1,1,1,1,0)\quad\quad \widetilde{D}^5(\tp;\B)=(0,0,0,0,0),$$
as agent~4 prefers buying four bundles in order to end up with $\date$ over buying the $\eld$.

Now suppose that the price of the bundle that only includes the $\eld$ falls, so we end up at bundle prices $\p'=(0.5,0.5,0.5,0.5,1)$. The bundled demands are then
$$\widetilde{D}^4(\tp';\B)=(0,0,0,0,1)\quad\quad \widetilde{D}^5(\tp';\B)=(1,0,0,0,1).$$
Agent~5 therefore increases their demand for $\apple$, while agent~4 switches from demanding $\apple$ and the three nontrivial bundles to demanding $\eld$. Hence, bundling goods with opportunities to sell reveals a bundle inconsistency that remains obscured if goods are only bundled with other goods.%
\hfill$\blacksquare$
\end{eg}

\end{document}